\documentclass[12pt]{article}

\title{Robust Mechanism Design on Networks with Externalities}
\author{Kohmei Makihara\footnote{Graduate School of Economics, the University of Osaka, 1-7, Machikaneyama, Toyonaka, Osaka 560-0043, Japan. Tel: +81-6-6850-6111. E-mail: k.makihara@econ.osaka-u.ac.jp}\\
}
\date{September 28, 2026}

\usepackage{cancel}
\usepackage{afterpage}  
\usepackage{color}
\usepackage[a4paper, hmargin={1in, 1in}, vmargin={1in, 1in}]{geometry}
\usepackage[dvipsnames]{xcolor}
\usepackage{fullpage}
\usepackage{amsmath}
\usepackage{amsthm}
\usepackage{amssymb}
\usepackage{multirow}
\usepackage{pdfpages}
\usepackage{color}
\usepackage{graphicx}
\usepackage[hidelinks]{hyperref}
\usepackage[english]{babel}
\usepackage[utf8]{inputenc}
\makeatletter
\DeclareRobustCommand{\cite}[1]{\citeauthor{#1}~(\citeyear{#1})}
\makeatother

\usepackage{comment}
\usepackage{appendix}
\usepackage{multirow,array}
\usepackage{tabularx}
\usepackage[T1]{fontenc}
\usepackage{lmodern}
\usepackage{tikz}
\usepackage[hang,small,bf]{caption}
\usepackage[subrefformat=parens]{subcaption}
\usepackage{chngpage}
\usepackage[normalem]{ulem}
\usepackage{enumerate}
\usepackage{newfloat,caption,float}
\usepackage{caption}
\usepackage{mathtools}
\usepackage[export]{adjustbox} 
\usepackage{xparse,etoolbox}
\usepackage{amssymb, amsthm, amsmath}

\DeclareMathOperator*{\argmax}{arg\,max}

\DeclareMathOperator{\BR}{BR}

\renewcommand{\bold}[1]{\ifmmode\mathbf{#1}\else\textbf{#1}\fi}

\theoremstyle{definition}
\newtheorem{theorem}{Theorem}
\newtheorem{proposition}{Proposition}
\newtheorem{definition}{Definition}
\newtheorem{lemma}{Lemma}
\newtheorem{corollary}{Corollary}

\makeatletter
\let\old@begintheorem=\@begintheorem
\def\@begintheorem#1#2{\itshape\old@begintheorem{#1}{#2}}
\makeatother

\newtheorem{manualtheoreminner}{Theorem}

\let\oldmanualtheoreminner\manualtheoreminner
\let\endoldmanualtheoreminner\endmanualtheoreminner
\RenewDocumentEnvironment{manualtheoreminner}{o}{%
  \IfNoValueTF{#1}{\oldmanualtheoreminner}{\oldmanualtheoreminner[#1]}%
  \itshape
}{%
  \endoldmanualtheoreminner
}

\newcounter{example}
\newenvironment{example}[1][]{\refstepcounter{example}\par
   \noindent \textbf{Example~\theexample. #1} }{}

\renewenvironment{itemize}{	%
	\begin{list}{\textbullet}{\topsep=1pt}
	\setlength{\parskip}{0mm}    %
	\setlength{\itemsep}{0mm}    %
}{
	\end{list}
}

\usepackage{enumitem}
\setlist[itemize,enumerate]{itemsep=0pt, topsep=0pt, parsep=0pt, partopsep=0pt}

\makeatletter
\newcommand{\@supervisor}{}
\newcommand{\supervisor}[1]{\renewcommand{\@supervisor}{#1}}
\newcommand{\@institute}{}
\newcommand{\institute}[1]{\renewcommand{\@institute}{#1}}
\makeatother

\DeclareRobustCommand{\erase}{\bgroup\markoverwith{\textcolor{red}{\rule[.5ex]{2pt}{0.4pt}}}\ULon}

\DeclareFloatingEnvironment[
  fileext = loe,
  listname = Examples,
  name = Example,
  placement = H,
  within = none,
  ]{ex}
\let\citet\textcite
\usepackage{csquotes}
\usepackage[backend=biber, style=apa, natbib=true, doi=false, url=false, isbn=false, eprint=false, related=false]{biblatex}
\theoremstyle{definition}
\newtheorem{remark}{Remark}
\let\oldremark\remark
\let\endoldremark\endremark
\renewenvironment{remark}{\oldremark\upshape}{\endoldremark}

\begin{document}

\maketitle

\begin{abstract}
We study how to allocate a good with positive externalities among agents in an information network, without monetary transfers. Each agent observes their own valuation and those of their neighbors. A principal seeks to allocate the good to the highest-valuation agent through a mechanism robust to agents' heterogeneous belief hierarchies, leveraging the network structure and the partial incentive alignment created by allocative externalities. We characterize the networks for which a robustly efficient mechanism exists: it exists if and only if at least two agents are connected to all others. Relaxing robustness to rule out weakly dominated strategies, one universally connected agent suffices. We generalize the results to a broader class of utility functions, show that robust welfare maximization is impossible unless the network is complete, and establish that an \textit{ex post incentive compatible} mechanism that is efficient under truthful reports exists if and only if no agent is isolated.
\vspace{\baselineskip}
\end{abstract} 

\section{Introduction}
Many allocation problems by the principal involve distributing resources to the most deserving agent, while the principal lacks this information. However, such information is often possessed locally by the friends or neighbors in the community. This type of local information is used in several environments to decide the allocation of resources: for instance, coordinators of aid programs use local information to target poor households in the community (e.g. \citet{alatas2012targeting}, \citet{alatas2016network}), or management of a company uses peer evaluation to assess worker performance to decide bonus or promotion allocation (e.g. \citet{pollack1996using}, \citet{fletcher2001performance}). 
In these settings, agents are often not indifferent to who receives the resource: households benefit when their community is accurately targeted by aid programs, since well-targeted programs increase trust in the system and future aid inflows, and coworkers benefit when the most productive colleague is promoted, since firm performance improves for the whole team. This creates positive externalities that give agents a partial incentive to report their local information truthfully, even without monetary transfers.

This paper formalizes this problem. A principal must allocate a good to the highest-valuation agent, but only agents observe valuations—their own, and those of the peers they are connected to in a network. The good generates a positive allocative externality that depends on the recipient's valuation: every agent benefits from a higher-valuation recipient, whether or not they are connected to that agent.\footnote{One interpretation of this specification is that the agent to whom the good is allocated produces a public good whose size increases in the recipient's productivity (i.e., valuation). Other agents then consume a fraction of this public good, given by their own externality factor.} The network therefore shapes only what agents know, not who benefits from the allocation. Our central question is the following: what are the necessary and sufficient network structures for which a mechanism exists that allocates the good to the highest-valuation agent in \textit{every} equilibrium, without restricting agents' beliefs?
Our results offer practical guidance for two decisions. First, they specify the mechanism a principal should use to elicit truthful reports from a given community. Second, they specify the network structure a policymaker should aim to create, when she can influence how information is shared, in order to obtain accurate information.

Because valuations are interdependent due to externalities—each agent's utility depends on the valuation of whoever receives the good—any mechanism induces a game of incomplete information, in which an agent's optimal message depends on his beliefs about others' valuations and strategies. Since agents differ in what they observe about the network, these beliefs need not be commonly known, and an agent's belief about others' beliefs, and so on, can differ as well. To design a mechanism whose efficiency does not depend on any such belief hierarchy, we adopt the solution concept of \textit{robust implementation} by \citet{bergemann2005robust, bergemann2009robust}, which requires the mechanism to produce an efficient allocation for any strategy profile that can arise as an interim equilibrium under \textit{any} belief hierarchy.\footnote{An interim equilibrium is a Bayes-Nash equilibrium where agents' beliefs do not necessarily admit a common prior (\citet{kunimoto2014interim}). The full characterization of such belief hierarchies follows \citet{harsanyi1967games}.}

Bringing this approach into our setting reveals why network structure and robust implementation combine naturally, and where the tension between them lies. Robust implementation demands efficiency across every belief hierarchy an agent could hold about what he cannot observe; an agent connected to more peers observes more valuations directly, which narrows the set of such hierarchies and makes the requirement easier to satisfy on his messages. But that same well-connected agent may exploit his position strategically, misreporting to secure the good for himself. Network structure therefore cuts both ways for eliciting true valuations. Robust implementation is demanding enough on its own that, absent further structure, it typically forces restrictive assumptions—private values, bounded interdependence, or monetary transfers—or fails outright.\footnote{See \citet{bergemann2009robust, ollar2017full}.}  Our central contribution is to show that the network's information structure, together with the positive externalities it carries, resolves this tension in the principal's favor across a broad class of environments—while also pinning down exactly where that resolution reaches its limit: not every network admits a robustly efficient mechanism, and characterizing which ones do is itself part of the contribution.

Our main finding, Theorem \ref{theo:robust_iff}, provides a tight characterization: such a mechanism exists if and only if at least two agents are connected to all others. The necessity part can be understood as follows. A non-central agent who believes an unobserved agent has the highest valuation regards allocating the good to that agent as payoff-maximizing---so any message profile that delivers this outcome is optimal for him, truthful or not, and hence an untruthful message profile may survive. 
Two universally connected agents overturn this: since each center observes the other, their reports can be directly compared, and the mechanism burns just enough of the good when they disagree that the higher-valuation center won't contest, yet the lower one still concedes. 
This result gives a concrete target to a policymaker able to shape the network: not one well-connected agent, but two, if a policymaker wants to guarantee that the outcome is always efficient, no matter what agents believe about each other.

Theorem \ref{theo:robust_iff}'s demanding robustness requirement restricts which networks it applies to, since it lets agents adopt untruthful strategies whenever they believe their own message won't affect the outcome. We therefore relax it by restricting attention to full-support beliefs over others' valuations and strategies, a criterion we call \textit{weakly robust efficiency}, equivalent to iterated elimination of weakly dominated strategies.\footnote{This elimination procedure lacks a Bayesian foundation, as \citet{borgers1994weak} notes. In the online appendix, we use the notion of rationalizability under approximate common knowledge (\citet{dekel1990rational}; \citet{borgers1994weak}; \citet{gul1996rationality}) to define \textit{admissibly robust efficiency}, and show a mechanism achieving it exists under the same network condition.} A weakly robustly efficient mechanism exists whenever at least one agent is connected to everyone — a weaker, more common network structure: many organizations have a single well-connected member, such as a supervisor or community leader, but few have two. Moreover, the star network — the sparsest network with this property — is itself a stable outcome of decentralized network formation, so this structure remains relevant even when a policymaker lacks the capacity to design the network.\footnote{See \citet{jackson1996strategic, bala2000noncooperative}.}

Theorem \ref{theo:robust_iff} characterizes the network structure required for a robust mechanism to exist; the natural question now is how much efficiency can be salvaged when a network fails to satisfy it. We tackle this question through what we call verifiable blocks: subnetworks with two internal hubs who observe everyone within them. Partitioning the network into such blocks yields a mechanism that robustly identifies the highest-valuation agent block by block. Then, when no single verifiable block spans the entire network, allocating the good to the best agent within the largest verifiable block guarantees a rank no worse than the number of agents outside that block. The guarantee therefore improves as the largest verifiable block grows: it recovers full efficiency when the block spans the whole network.

To investigate how the model behaves in other environments, we generalize both the utility function and the principal's objective. We introduce \textit{linear-in-valuations} utility functions, and show our main results extend whenever every agent's largest externality is for the agent the principal wants to allocate to — a condition flexible enough to let the externality factor vary with the recipient, i.e., a single agent's externality factor can differ depending on who obtains the good. This alignment need not hold when the principal maximizes utilitarian welfare instead: because the externality factor is heterogeneous across recipients, the welfare-maximizing allocation need not maximize the utility of every agent who does not obtain the good.
In that case, a robustly efficient mechanism exists only under the complete network. The contrast underscores that incentive alignment between principal and agents is the key feature of the model enabling robust implementation.

We also examine our model under \textit{ex post full implementation}: a weaker benchmark that fixes agents' beliefs about others' behavior at the equilibrium strategy itself, rather than allowing the arbitrary belief hierarchies that drive our robustness requirement. We evaluate \citet{bergemann2008ex}'s two necessary conditions — ex post incentive compatibility (EPIC) and ex post monotonicity (EM) — in our setting. Both are easy to satisfy: a mechanism that is efficient under truthful reports and satisfies EPIC exists if and only if no agent is isolated,\footnote{A mechanism that burns the good at every message profile is trivially EPIC on any network. This is why we restrict attention to mechanisms that are efficient under truthful reports.} and EM holds for any network under our utility function. The contrast is stark: ex post implementation asks only that no agent be isolated, while our results require one or two universally connected agents — showing that our network conditions are driven by robustness to belief hierarchies, not by incentive compatibility alone.

\subsection*{Related literature and contribution}

This paper contributes to several strands of literature. First, it advances recent work on mechanism design in information networks and peer selection problems. \citet{bloch2022friend} study how a planner can construct a truthful ranking of agents by an ordinal characteristic (e.g. ability or need), using only the reports agents can make locally about their friends in a social network; agents care only about having a higher rank themselves. Their solution concept is ex post incentive compatibility, which does not require full implementation. \citet{baumann2025robust} studies a similar prize-allocation problem and achieves full implementation in Bayesian Nash equilibrium under a commonly known prior distribution over values, made possible by a limit on how much agents can lie (evidence). \citet{niemeyer2024optimal} also study a mechanism design problem where agents hold peer information; their agents likewise have private values, and dominant-strategy incentive compatibility gives them robustness to beliefs for free, though their practical mechanisms are only approximately optimal in the large-network limit. \citet{bloch2025peer} study a peer selection problem in which agents care about who wins even when it is not themselves, i.e. a setting with interdependent values.  Their model differs from ours in what the network represents: rather than encoding who observes whom, it encodes each agent's preferences over who is selected by the planner, while every agent already knows which agents are deserving of the good, the information the planner is trying to elicit. They characterize when a dominant-strategy incentive-compatible and efficient mechanism exists. Our central contribution is to show that a specific form of interdependent values governed by positive allocative externalities, mediated through the right network structure, lets a designer achieve full robust implementation — efficiency under every belief hierarchy, with no common prior and no restriction on how agents' values depend on one another — a combination none of these four papers deliver.

Second, this paper advances implementation theory and mechanism design in environments with interdependent valuations, particularly in contexts without monetary transfers. In these settings, agents' utility depends on the valuations of others, with each agent's valuation typically unknown to the others, creating an incomplete information game. Previous work on this topic has largely focused on allocations with monetary transfers (\citet{cremer1985optimal}, \citet{jehiel2001efficient}) and, more recently, on allocations without transfers (\citet{bhaskar2019resource}, \citet{goldlucke2020multiple}, \citet{feng2023limits}). Our study extends this literature by examining how specific network structures can ensure the existence of an efficient mechanism.
In this context, outcomes rely heavily on common knowledge assumptions that shape agents' beliefs. Assuming a common prior as in standard Bayes-Nash analysis is one way to make the analysis feasible, but this does not fully capture the heterogeneous higher order beliefs that agents may possess.\footnote{The criticism of excessive reliance on common knowledge assumption is referred to as the \textit{Wilson doctrine} and has garnered significant attention in implementation theory.}

To address this point, \citet{bergemann2005robust, bergemann2009robust} introduce the concept of robust implementation, and establish necessary and sufficient conditions for the existence of such mechanisms; however, they take a specific structure of the utility function,\footnote{They assume that the utility of an agent depends only on the {\it aggregate} of the valuations, i.e. there exists an aggregator $h_i(\cdot)$ which gives a scalar depending on the realized valuation profiles.} and, due to the allocative externalities that are present in our model, their framework does not cover our setting. \citet{ollar2017full} address a related problem by providing conditions for an efficient monetary transfer scheme to exist under general restrictions on agents' beliefs. Their key tool is the notion of a \textit{moment condition}.\footnote{This is represented as a function $L_i$ of others' types whose conditional expectation, given agent $i$'s own type, is commonly agreed upon across all admissible beliefs.} Intuitively, a moment condition encodes partial common knowledge about the type-generating process that every agent agrees on regardless of their particular belief. The designer exploits this shared knowledge to adjust the canonical VCG transfers, weakening the strategic externalities in the mechanism just enough to ensure that truthful revelation is the unique rationalizable outcome.\footnote{Robust implementation is also discussed in the context of dynamic mechanism design or virtual implementation. See, for instance, \citet{bergemann2009virtual}, \citet{penta2015robust}, \citet{muller2016robust, muller2020robust}, \citet{battigalli2026monotonicity}.}

Our contribution is to present a novel approach to achieving robust efficiency by leveraging network structures where agents have information about some of their peers, without imposing any belief restrictions and without relying on monetary transfers. Instead, we focus on the network structure and the positive externalities that align the incentives of the principal and the agents. We note that we are not the first to utilize externalities as a means of incentive alignment. For instance, \citet{bhaskar2019resource} discuss the conditions for a second-best optimal mechanism and, using positive externalities, propose an optimal mechanism in a Bayes-Nash equilibrium under a common prior, where agents' valuations are private information  which, in our setting, amounts to considering that the network is empty. Our value-added regarding the use of the externalities is the following: by introducing the information network into such context, not only the first best is possible for some networks, but also with stronger solution concept which does not assume any common knowledge.

\subsection*{Outline}
The paper is organized as follows. Section \ref{sec:model} presents the model, Section \ref{sec:rob_eff} characterizes when a robustly efficient mechanism exists, and Section~\ref{sec:weak_rob_eff} studies a weaker notion of robustness that requires only one universally connected agent. Section~\ref{sec:general_networks} extends the analysis to networks without such structure, Section~\ref{sec:discussion} discusses a broader class of utility functions and the welfare-maximization problem, and applies the concept of ex post implementation.

\section{Model}\label{sec:model}
A principal has one unit of a divisible good to allocate among $n\geq 2$ agents. Let $N = \{1,\cdots,n\}$ be the set of all agents in the community. We denote the amount of good allocated to agent $i$ as $x_i$, and the allocation profile as ${\bf x} = (x_1,\cdots,x_n) \in X$. We require that for all $i \in N$, $x_i \geq 0$, and $\sum_{i \in N}x_i \leq 1$. This implies that some of the good may not be allocated, allowing for money burning, i.e. the principal can choose not to allocate a portion of the good. The divisibility of the good can be interpreted in two ways: either as the inherent nature of the good itself (e.g., if the good to be allocated is a budget sum) or as the probability of allocating an indivisible good to an agent (e.g., the allocation of a single promotion or prize).

Each agent derives utility from the allocation profile. The (ex-post) utility function for agent $i$ is defined as follows:
$$
u_i({\bf x};{\bf v}) = (1 - \alpha_i)v_i x_i + \alpha_i \sum_{j \in N}v_j x_j
$$
where $v_i \in [0,1]$ is the valuation of agent $i$, drawn from an unknown distribution, and $\alpha_i \in (0,1)$ is $i$'s externality factor.\footnote{If $\alpha_i = \alpha$ for all $i$, the utility function is equivalent to the one employed by \citet{bhaskar2019resource}.}
The term $v_i$ represents the constant marginal utility of the good allocated to agent $i$, indicating either the agent's productivity\footnote{Under this interpretation, an agent who receives the good generates $v_i$ units of output with positive externalities.} or the degree to which each agent values the good. We denote the valuation profile by ${\bf v} = (v_1,\cdots, v_n)$, and let $\mathbf{v}_{I} := (v_{i})_{i \in I}$ for $I \subseteq N$. We call the agent who has the highest valuation the highest agent, and denote him by $i^{*} := i^{*}(\mathbf{v})$. We assume that $v_i \neq v_j$ for any $i \neq j$. 

The principal's objective is to allocate the entire good to agent $i^{*}$. For instance, in the context of task allocation, this means assigning the task to the most productive agent. In the case of a simple allocation of a single good, the principal aims to allocate it to the agent who values it the most. Let $\mathbf{e}_{i}$ denote the vector with $1$ at the $i$ entry and $0$ otherwise. We say that the allocation ${\bf x}$ is efficient for a given realization ${\bf v}$ if ${\bf x} = \mathbf{e}_{i^{*}}$. 

Given the principal's objective, the externality factor $\alpha_i$ represents how closely aligned the incentives are between the principal and the agents. As $\alpha_i \to 1$, the utility approaches $\sum_{j \in N}v_jx_j$ which is maximized when we allocate the entire good to the agent with the highest valuation. In this case, the efficient allocation maximizes the utility function of each agent. Conversely, as $\alpha_i$ decreases, the weight of the first term in the utility function increases, suggesting that the significance of each agent's own allocation becomes more pronounced.%

The externality factor $\alpha_i$ can be interpreted in various ways depending on the context. For instance, consider a scenario where the good being allocated is the assignment of a project manager to a team member. Once the project manager is assigned and the project is completed, utility is realized based on the individual reward given only to the manager, alongside a group reward distributed to all team members. In this case, $\alpha_i$ represents the relative magnitude of the group reward compared to the individual reward. %
In another context, when the principal aims to allocate the good to the individual who values it most, $\alpha_i$ can be seen as a measure of altruism of agent $i$, where he receives part of the material benefit of agent $j$ being $v_j x_j$. Furthermore, considering that the second term in the utility function represents total welfare, $\alpha_i$ can also be interpreted as a preference for efficiency, as discussed in various studies in experimental economics.\footnote{See for instance \citet{charness2002understanding} or \citet{engelmann2004inequality}.}

We assume that the principal does not know the valuation of any agent, while each agent is aware of his own valuation as well as the valuations of his neighbors. 
Let $N_i \subseteq N$ denote the set of agents whose valuations are known by agent $i$, that is, it includes agent $i$'s neighbors and himself. We assume that the information is mutual, meaning that if agent $i$ knows agent $j$'s valuation, then agent $j$ also knows agent $i$' one. This forms a network $\mathbf{G} = (g_{ij})_{i,j \in N}$ such that $g_{ij} = g_{ji} = 1$ if $i \in N_{j}$, and $g_{ij} = g_{ji} = 0$ otherwise.
The information possessed by agent $i$ is described by $\theta_i(\mathbf{v}) := (v_{j})_{j \in N_{i}}$, and we call it agent $i$'s type.\footnote{Strictly speaking, a type profile should be defined as a function of the realization of the valuations and the set of neighbors, i.e. $\theta_i(\mathbf{v}, N_{i}) := \mathbf{v}_{N_{i}}$, but for the sake of simplicity, we omit the notation for the dependence of the type on the set of neighbors and write $\theta_i(\mathbf{v}) := \theta_i({\bf v}, N_{i})$, or simply $\theta_i$ when there is no risk of confusion.} Throughout, we write $\mathbf v = (v_j)_{j \in N}$ for the true, realized valuation profile, and $\mathbf v' = (v'_j)_{j \in N}$ for a generic profile that an agent may consider possible.
The network $\mathbf{G}$ and the externality factor $(\alpha_i)_{i \in N}$ are common knowledge among agents and the principal. The structure of the utility function, combined with the fact that an agent does not have complete knowledge of everyone's valuations, implies that an agent lacks full information about his utility function ex-ante.%

\subsection*{Mechanism}
The principal's objective is to achieve an efficient allocation for any realization of valuation profile ${\bf v}$, which is unknown for her. To accomplish this, the principal proposes and commits to a mechanism, defined by the message space and the allocation rule. A mechanism is a pair $(M, g(\cdot))$, where $M := \prod_{i \in N}M_{i}$ with $M_i$ being the message space for agent $i$, and $g = (g_{1}, \cdots, g_{n})$ where $g: M \to X$ is the allocation rule. 

From the agents' perspective, each mechanism induces a game of incomplete information, in which the strategy set is his message space $M_i$, and the ex-post payoff function is defined as follows:
\begin{align*}
\pi_i(m_i,m_{-i};{\bf v}) = u_i(g(m_i,m_{-i});{\bf v})
\end{align*}
We consider a direct mechanism such that agents report the valuation of their neighbors and themselves.\footnote{Some studies on implementation theory and mechanism design take the indirect mechanism approach to eliminate inefficient allocation profiles from the set of equilibria of the game induced by the mechanism (see for instance \citet{palfrey1989mechanism} or \citet{bergemann2008ex}). However, this often involves an augmented mechanism that uses variants of integer games as a device to eliminate inefficient allocations. Not only are integer games considered unrealistic to implement, but they might also lack any equilibrium due to the message space being infinite.} The message space of agent $i$ is $M_i := [0,1]^{N_i}$, and we write $V_{-i} := [0,1]^{N \setminus \{i\}}$ for the set of valuation profiles of the other agents. We denote the message of agent $i$ about the valuation of agent $j$ by $m_{ij}$, so that $m_{i} = (m_{ij})_{j \in N_{i}}$. For simplicity, we restrict the attention to direct mechanisms with $g(\theta(\mathbf{v})) = \mathbf{e}_{i^{*}}$.\footnote{It means that if everyone is truthful, then the mechanism allocates the good to the highest agent.}
We say that a message profile $\mathbf m = (m_i)_{i\in N}$ is consistent if, for every $j \in N$ and every $i,i' \in N_{j}$, $m_{ij} = m_{i'j}$---that is, every pair of agents who observe agent $j$ report the same valuation for him---and inconsistent otherwise.

\section{Beliefs and Robustly efficient mechanism}\label{sec:rob_eff}
In this paper, we take the approach of robust implementation proposed by \citet{bergemann2005robust}. This robust concept ensures that the mechanism leads to an efficient allocation for any message profile that could be supported as an interim equilibrium under some belief hierarchy. In other words, it is robust against any possible configuration of beliefs and higher-order beliefs among agents. 
Following the literature, we use the iterative elimination of strategies that are never a best response, which is equivalent to assessing interim equilibria across all belief spaces.

To introduce this approach, we first define the expected payoff and the best responses. Agent $i$'s expected payoff depends on his conjecture conditional on his type, which is a joint probability distribution $\mu_i(\cdot\mid\theta_i)$ over the messages and valuations of the other agents. 
Given his conjecture $\mu_i(\cdot\,|\,\theta_i)$, the expected payoff when agent $i$ adopts a message $m_i$ is
$$
\mathbb{E} \pi_i(m_i; \theta_i, \mu_i) = \int_{M_{-i} \times V_{-i}}u_i(g(m_i, m_{-i});{\bf v})d\mu_i(m_{-i},v_{-i}|\theta_{i})
$$
and the best responses are
$$
\BR_{i}(\theta_i,\mu_i) = \argmax_{m'_i \in M_i}\,\mathbb{E} \pi_i(m'_i; \theta_i, \mu_i)
$$

The concept of robust efficiency is defined based on the iterative elimination of never-best responses for each agent given his information.\footnote{This concept is a special case of {\it $\Delta$-rationalizability} proposed by \citet{battigalli2003rationalization}.} Let $S^0_i = M_{i}$, and for $k \geq 1$ define
\begin{align*}
  E_i^{k-1}(\theta_i) := \Big\{(m_{-i},v'_{-i}) \in M_{-i}\times V_{-i} \;\Big|\; &(v'_j)_{j\in N_i \setminus \{i\}} = (v_{j})_{j\in N_i \setminus \{i\}} \text{ and } \\
  &m_{-i}\in S_{-i}^{k-1}(\theta_{-i}(v_i,v'_{-i}))\Big\}
\end{align*}
$$
S^k_i(\theta_i) := \Big\{ m_i \in M_i \;\Big|\; \exists \mu_i \in \Delta(M_{-i}\times V_{-i}) \text{ s.th. } \mu_i\big(E_i^{k-1}(\theta_i)\mid\theta_i\big) = 1 \text{ and } m_i \in \BR_i(\theta_i,\mu_i) \Big\}
$$
The set $E_{i}^{k-1}(\theta_i)$ collects pairs of message and valuation profiles of others which agent $i$ regards as possible after $(k - 1)$-th round of iterations, whereas $S^k_i(\theta_i)$ collects the messages of agent $i$ that survives the $k$-th round of iteration.
Within $E_{i}^{k-1}(\theta_i)$, the condition $(v'_j)_{j\in N_i \setminus \{i\}} = (v_{j})_{j\in N_i \setminus \{i\}}$ requires the valuation profile to be consistent with what agent $i$ observes,\footnote{Note that this observational-consistency restriction corresponds to the belief restriction $C_i^{\mathcal{B}}(\theta_i)$ of $\mathcal{B}$-Rationalizability in \citet{ollar2017full}: the possible beliefs of an agent are exactly those consistent with what he observes. We write $C_i(\theta_i)$ for the set of such conjectures, i.e. $\mu_i \in \Delta(M_{-i}\times V_{-i})$ assigning probability 1 to $v'_{-i}$ with $(v'_j)_{j\in N_i \setminus \{i\}} = (v_{j})_{j\in N_i \setminus \{i\}}$.} while $m_{-i}\in S_{-i}^{k-1}(\theta_{-i}(v_i,v'_{-i}))$ requires the other agents' message profiles to have survived the previous round.

The definition of $S^k_i(\theta_i)$ then states the condition for a message of agent $i$ to survive the $k$-th round of iteration: $m_{i}$ survives if it is a best response to some conjecture that assigns the positive probability only to the profiles in $E_{i}^{k-1}(\theta_{i})$. Equivalently, a message is eliminated if it is never a best response to any such conjecture. 

Let $S_{i}(\theta_i) = \cap_{k \geq 0}S^k_i(\theta_i)$. Moreover, we say that a message profile ${\bf m}$ is rationalizable if ${\bf m} \in \Pi_{i \in N}\,S_{i}(\theta_i)$.
\begin{definition}[Robust efficiency]
Given a network $\mathbf{G}$ and $(\alpha_{i})_{i \in N}$, the allocation rule $g$ is robustly efficient if for all ${\bf v}$, $g({\bf m})$ is efficient for all $\mathbf{m}$ which are rationalizable.
\end{definition}
$S^k_i(\theta_i)$ is the set of all messages that survive the $k$-th round of elimination. %
We can see that for all $k$, $S^k_i(\theta_i) \subseteq S^{k-1}_i(\theta_i)$ which guarantees that this process ends after certain number of rounds, i.e. there exists $k^*$ such that $S^{k^*}_i(\theta_i) = S^{k^*+1}_i(\theta_i)$ for all $i$.\footnote{We can see that by redefining this process in terms of an operator\vspace{-2mm}
$$
b_i(S,\theta_i) = \left\{ m_i \in M_i \,|\,\exists \mu_i \in C_{i}(\theta_i) \text{ s.t. } m_i \in \BR_i(\theta_i,\mu_i) \text{ and } \mu_i\left(\{(m_{-i},v_{-i})\,|\,m_{-i} \in S_{-i}(\theta_{-i})\}\right) = 1 \right\} \vspace{-2mm}
$$
which is monotone in the set inclusion ordering, hence to which the Tarski's fixed point theorem is applicable.}\textsuperscript{,}\footnote{Note that if the network is empty, i.e. no one is connected to anyone, then the set $S_i(\theta_i)$ corresponds to $S_{i}^{\mathcal{M}}(\theta_i)$ of \citet{bergemann2009robust}, known as belief-free environments.}
It is known that a message profile is in $\Pi_{i \in N}S_i(\theta_i)$ if and only if it can be played as an interim equilibrium in certain beliefs of agents.\footnote{See Proposition 1 of \citet{bergemann2011robust}.} Thus, robust efficiency requires that, regardless of the agents' beliefs\textemdash including all higher-order beliefs\textemdash the mechanism consistently results in an efficient allocation at an interim equilibrium. This invariance to belief variations is the essence of its robust nature.	

To find a robustly efficient mechanism, it is necessary that all message profiles surviving the iterative elimination process yield an efficient allocation under the mechanism. This requires not only ensuring that among the surviving message profiles, there exists at least one that leads to an efficient allocation (partial implementation), but also that all inefficient message profiles are removed from the set of possible outcomes (full implementation).
\begin{remark}
Penalizing inconsistent messages alone does not guarantee full implementation. For instance, with the mechanism that allocates nothing whenever agents disagree, a consistent but untruthful message profile—each claiming that the same agent has the highest valuation, even though he does not—survives iterative elimination, since for each agent, such an untruthful message is a best response to the others' messages: a deviation from this profile results in no allocation to anyone, since the message profile is then no longer consistent. Achieving robust efficiency therefore requires that the mechanism also creates strict incentives to report truthfully when messages are consistent, not merely when they conflict.
\end{remark}

The primary challenge here stems from the lack of monetary transfers. In settings with monetary transfers, the principal could design a transfer scheme to reward agents individually, ensuring that one agent's reward does not impact the utility of others. However, without monetary transfers, the principal's only available punitive tool is money burning. This approach is inherently limited, as money burning acts as a collective punishment. Because of the positive externalities in our setting, agents prefer that the good is allocated to someone\textemdash even if it is not themselves\textemdash rather than being destroyed.

Despite this difficulty, the following lemma shows that in the 2-agents case, there exists a simple allocation rule which rewards the agent with the higher valuation and punishes the lower agent, so that untruthful consistent message profiles are eliminated from the set of rationalizable message profiles.
\begin{lemma}\label{lem:two_agents}
Assume that $N = \{i,j\}$ and they are connected. The direct mechanism such that 
\begin{align*}
  g(\mathbf{m}) = \begin{cases}
    (1, 0), &\text{ if } m_{ii} \geq m_{ij} \text{ and } m_{ji} \geq m_{jj}\\
    \left(\beta_{ij}, \beta_{ji}\right), &\text{ if } m_{ii} \geq m_{ij} \text{ and } m_{ji} < m_{jj}\\
    (0,0), &\text{ if } m_{ii} < m_{ij} \text{ and } m_{ji} \geq m_{jj}\\
    (0,1), &\text{ if } m_{ii} < m_{ij} \text{ and } m_{ji} < m_{jj}
  \end{cases}
\end{align*}
where 
$$
(\beta_{ij}, \beta_{ji}) = \bigg(\frac{\alpha_i(1 - \alpha_j)}{1 - \alpha_i \alpha_j}, \frac{\alpha_j(1 - \alpha_i)}{1 - \alpha_i \alpha_j}\bigg)
$$
is robustly efficient.
\end{lemma}
\begin{proof}
All proofs are relegated to the appendix.
\end{proof}
To see why this split works, compare agent $i$'s utility, $u_i(\mathbf x;\mathbf v) = v_ix_i + \alpha_i v_j x_j$, under two messages when agent $j$ claims to be highest: claiming himself highest too (triggering the punishment split $(\beta_{ij},\beta_{ji})$) versus conceding (letting $j$ take the whole good). The utility difference is
$$
  \big[v_i\beta_{ij} + \alpha_i v_j \beta_{ji}\big] - \alpha_i v_j = v_i\beta_{ij} - \alpha_i v_j(1-\beta_{ji}) = v_i\beta_{ij} - v_j\beta_{ij} = \beta_{ij}(v_i - v_j),
$$
where the last step uses the identity $\alpha_i(1-\beta_{ji}) = \beta_{ij}$ (immediate from the formulas for $\beta_{ij},\beta_{ji}$). Since $\beta_{ij}>0$, this difference is strictly positive exactly when $v_i>v_j$ and strictly negative exactly when $v_i<v_j$. This exact cancellation is what pins down $\beta_{ij}$.\footnote{Note that this mechanism is not robustly efficient if we allow for valuation profiles where $v_{1} = v_{2}$. In such cases, both agents are indifferent between allocating the entire good to the other agent and receiving the partial allocation $(\beta_{12}, \beta_{21})$. Consequently, untruthful but consistent message profiles can survive the iterative elimination process. This illustrates the fundamental difficulty of dealing with ties in valuations, which is why we focus our analysis on the generic case where all agents have strictly distinct valuations.}

Assuming agent $j$ claims to be highest, a higher $\alpha_i$ raises the externality agent $i$ sacrifices by claiming himself highest instead of conceding: conceding delivers him the full externality $\alpha_iv_j$, while claiming himself highest delivers only the fraction $\alpha_iv_j\beta_{ji}$ he captures through $j$'s reduced share $\beta_{ji}$. To keep agent $i$ willing to claim himself highest when he truly is, his private allocation $\beta_{ij}$ must rise to cover this larger gap — and fall as $\alpha_i$ falls.
As $\alpha_i\to0$, agent $i$'s utility from conceding is zero since he does not receive externalities, so $\beta_{ij}$ must also converge to zero: otherwise agent $i$ would prefer to claim himself highest even when his own valuation is lower. 
At $\alpha_i\to1$, agent $i$'s utility depends only on the total value delivered ($u_i\approx v_ix_i+v_jx_j$), so he always prefers that the good is allocated to the highest agent, and has no incentive to claim himself as the highest when he is not.\footnote{This underscores the interpretation of $\alpha_i$ as the degree of incentive alignment between the principal and agent $i$: at $\alpha_i\to0$ incentives are not aligned at all, and no split can incentivize agent $i$ to tell the truth, whereas at $\alpha_i\to1$ incentives are fully aligned. Note also that as $\alpha_i$ or $\alpha_j$ increase, the amount of money burning decreases. This also illustrates that $\alpha_i$ serves as a measure of the degree of incentive alignment between the principal and the agents, since the principal can elicit the true information with less money burning.}

Using this two-agent mechanism, we can show that it suffices to have two agents within the network who are connected to all other agents. Furthermore, we can demonstrate that this particular network structure is also necessary for the existence.
\begin{theorem}\label{theo:robust_iff}
A robustly efficient mechanism exists if and only if there are at least two agents who are connected to everyone.
\end{theorem}

The intuition of the necessity is as follows: with only one central agent, every other agent fails to observe someone, so—as the example below shows—he can always rationalize an untruthful message by believing the good will go to that unobserved agent regardless of what he reports.\vspace{3mm}
\begin{example}[(Necessity of two central agents)]\label{ex:necessity}
Consider the network in Figure \ref{fig:one_center}, where only agent 1 is connected to everyone, and let the true valuation profile satisfy \(v_2>v_1>v_3>v_4\). We show that no robustly efficient mechanism can exist.
\begin{figure}[htbp]
\centering
\begin{tikzpicture}[every node/.style={circle, draw, fill=gray!25, inner sep=1.5pt, font=\small}]
  \node (4) at (0,2.2) {4};
  \node (1) at (0,1.1) {1};
  \node (2) at (-0.75,0) {2};
  \node (3) at (0.75,0) {3};
  \draw (1) -- (4) (1) -- (2) (1) -- (3) (2) -- (3);
\end{tikzpicture}
\caption{Four-agent network where $v_2>v_1>v_3>v_4$ with only one agent connected to everyone. }\label{fig:one_center}
\end{figure}
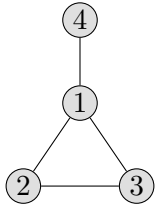
The key observation is that any non-central agent fails to observe at least one other agent. For instance, agent 2 does not observe agent 4. Hence, agent 2 may hold a belief under which agent 4 has the highest valuation and all other agents send messages leading the mechanism to allocate the good to agent 4. Under such a belief, agent 2 has no incentive to correct a message that ranks agent 1 above agent 2, since he still expects the good to go to agent 4. The same reasoning applies to agent 3.

As a result, the message profile in which the non-central agents support agent 1 as the highest remains rationalizable. Agent 1 can then also report himself as highest, so the mechanism allocates the good to agent 1, even though agent 2 is truly highest. This contradicts robust efficiency.\vspace{3mm}
\end{example}

This logic extends beyond this example to any network with fewer than two universally connected agents. If fewer than two agents are connected to everyone, then at least $n-1$ agents each fail to observe someone else, just as agents 2 and 3 above fail to observe agent 4, and agent 4 fails to observe agents 2 and 3. Each such agent can believe that the agent he does not observe has the highest valuation, so that the allocation he expects—the entire good going to that unobserved agent—is already utility-maximizing for him. Since he has no incentive to induce any other allocation, he also has no incentive to correct his report of the ranking among the agents he does observe: an untruthful report of that ranking is then a best response too, as it leaves this already-optimal expected outcome unchanged. Robust efficiency therefore requires that at least two agents be connected to everyone. This is exactly the target identified in the introduction: a policymaker able to shape the network should aim not for one well-connected agent, but two.

Sufficiency builds on Lemma \ref{lem:two_agents}: two agents connected to everyone already know every agent's true valuation, so applying Lemma \ref{lem:two_agents}'s punishment device to these two centers alone makes truthful reporting dominant for each of them, and their reports pin down the efficient allocation regardless of what anyone else claims. We construct such a mechanism explicitly. Let $\gamma^{ij} \in X$ be the allocation profile such that $\gamma^{ij}_{i} = \beta_{ij}$, $\gamma^{ij}_{j} = \beta_{ji}$, and $\gamma^{ij}_{k} = 0$ for $k \neq i,j$.
\vspace{2mm}\\
\textbf{The two-center mechanism}\\
Let $C \subseteq N$ be the set of agents who are connected to everyone. Choose two agents from $C$, and without loss of generality let them be agents $1$ and $2$. Let
\begin{align*}
  h_{i}(m_{i}) &= \min\big\{j \in N_{i} : m_{ij} = \max_{k \in N_{i}}m_{ik} \big\}\\
  \hat{h}_{i}(m_{i}) &= \min\big\{j \in N_{i}\setminus\{i\} : m_{ij} = \max_{k \in N_{i}\setminus\{i\}}m_{ik} \big\}
\end{align*}
In words, $h_{i}(m_{i})$ is the agent who has the lower index among the neighbors of $i$ who are ranked as the highest by $i$, and $\hat{h}_{i}(m_{i})$ is of the same definition but excluding agent $i$ himself. Let the allocation rule $g(\mathbf{m}) = g(m_{1}, m_{2})$ be the following:
\begin{enumerate}
\item If $h_{i} \neq i$ for all $i \in \{1,2\}$ and $h_{1} \neq h_{2}$, then $g_{h_1}(m_{1},m_{2}) = g_{h_2}(m_{1},m_{2}) = \frac{1}{2}$. If $h_1 = h_2 \equiv h^*$, then $g_{h^{*}}(m_{1},m_{2}) = 1$.
\item If $h_{1} = 1$ and $h_{2} = 2$, then $g(m_{1},m_{2}) = \gamma^{12}$.
\item If $h_{i} = 1$ for all $i \in \{1,2\}$, then $g(m_{1},m_{2}) = \mathbf{e}_{1}$. If $h_{i} = 2$ for all $i \in \{1,2\}$, then $g(m_{1},m_{2}) = \mathbf{e}_{2}$.
\item Otherwise,
\begin{itemize}
  \item[-] if $h_{1} \neq 1,2$ and $h_{2} = 2$, then $g_{h_1}(m_{1},m_{2}) = m_{12}/m_{1h_{1}} + \epsilon_{1}(\mathbf{m})$
  \item[-] if $h_{1} = 1$ and $h_{2} \neq 1,2$, then $g_{h_2}(m_{1},m_{2}) = m_{21}/m_{2h_{2}} + \epsilon_{2}(\mathbf{m})$
\end{itemize}
with $\epsilon_{1}(\mathbf{m}) > 0$ and $\epsilon_{2}(\mathbf{m}) > 0$ small enough so that $m_{12}/m_{1h_{1}} + \epsilon_{1}(\mathbf{m}) < 1$ and $m_{21}/m_{2h_{2}} + \epsilon_{2}(\mathbf{m}) < 1$,
and $g_i(m_{1},m_{2}) = 0$ for all other $i$.
\end{enumerate}\vspace{2mm}

The mechanism uses only the two central agents' messages, since together they already know every agent's true valuation. When their reports disagree on who is highest between themselves, the punishment split $(\beta_{12},\beta_{21})$ from Lemma \ref{lem:two_agents} applies, making truthful reporting dominant for exactly the same reason as before.  The remaining case is where some other agent has a higher valuation than both centers: the lower-valuation center still prefers that this outside agent receive the good rather than let the higher-valuation center obtain it, so he has no incentive to falsely report the latter as highest. This, in turn, leaves the higher-valuation center with no incentive to claim himself highest either, since his claim would no longer be supported by the lower-valuation center—so he reports the true highest agent instead, to avoid the money burning.\footnote{The money burning in this case can be driven by point 4 of the mechanism.} Together, these two forces make only honest reporting rationalizable for both centers, which is what pins down the efficient allocation.

\section{Mechanism with weak robustness}\label{sec:weak_rob_eff}
As Example \ref{ex:necessity} illustrates, the impossibility in Theorem \ref{theo:robust_iff} arises because a non-central agent can always believe the good will go to an agent he does not observe, making an untruthful message appear payoff-maximizing to him even though it is not. This section shows how to rule this out, for networks with at least one central agent, by relaxing robustness: in each round of the elimination process, we also remove messages that are weakly dominated. Because the central agent observes everyone, he has no unobserved agent to hide behind, so we can construct the mechanism such that his untruthful messages are weakly dominated and are eliminated already in the first round. Once non-central agents can infer that the center reports truthfully, their own incentive to misreport disappears as well, breaking the impossibility.

Let \(T_i^0(\theta_i) := M_i\) and the true valuation be $\mathbf{v}$. For $k\ge 1$, define
\begin{align*}
  F_i^{k-1}(\theta_i) := \Big\{(m_{-i},v'_{-i}) \in M_{-i}\times V_{-i} \;\Big|\; &(v'_j)_{j\in N_i \setminus \{i\}} = (v_{j})_{j\in N_i \setminus \{i\}} \text{ and } \\
  &m_{-i}\in T_{-i}^{k-1}(\theta_{-i}(v_i,v'_{-i}))\Big\}
\end{align*}
\begin{align*}
T_i^k(\theta_i)
:=
\bigg\{
&m_i\in M_i\;|\;\exists \mu_i\in \Delta(M_{-i}\times V_{-i})\text{ s.th. }\\
&\mu_i(F_i^{k-1}(\theta_i)\mid \theta_i)=1,
\;\\
&\mu_i(U\mid\theta_i) > 0 \text{ for every nonempty relatively open } U \subseteq F_i^{k-1}(\theta_i),
\text{ and }\\
&m_i\in \BR_i(\theta_i,\mu_i)
\bigg\}.
\end{align*}
$F_i^{k-1}(\theta_i)$ plays exactly the role of $E_i^{k-1}(\theta_i)$ in Section~\ref{sec:rob_eff}, built from $T_{-i}^{k-1}$ instead of $S_{-i}^{k-1}$: the condition $(v'_j)_{j\in N_i \setminus \{i\}} = (v_{j})_{j\in N_i \setminus \{i\}}$ requires the valuation profile to be consistent with what agent $i$ observes, while $m_{-i}\in T_{-i}^{k-1}(\theta_{-i}(v_i,v'_{-i}))$ requires the opponents' message profile to have survived the previous round.

The definition of $T_i^k(\theta_i)$ mirrors that of $S_i^k(\theta_i)$, with one additional requirement: rather than merely assigning probability 1 to $F_i^{k-1}(\theta_i)$, the conjecture must have full support on it, that is, it must assign positive probability to every nonempty relatively open subset of $F_i^{k-1}(\theta_i)$.\footnote{We state full support in this form rather than as $\operatorname{supp}(\mu_i) = F_i^{k-1}(\theta_i)$ because a support is always a closed set, while $F_i^{k-1}(\theta_i)$ need not be closed.} This is the only difference between the two constructions, and it is what turns elimination of never-best responses into elimination of weakly dominated messages, in the sense of Lemma \ref{lem:WD_elimination} below.

Let $T_i(\theta_i)=\bigcap_{k\ge 0} T_i^k(\theta_i)$.
We say that a message profile ${\bf m}$ is weakly rationalizable if ${\bf m}\in \prod_{i\in N} T_i(\theta_i)$. 

\begin{definition}[Weakly robust efficiency]
Given the network ${\bf G}$ and $(\alpha_{i})_{i \in N}$, the allocation rule $g$ is weakly robustly efficient if, for all ${\bf v}$, $g({\bf m})$ is efficient for every weakly rationalizable message profile ${\bf m}$.
\end{definition}

Since message and valuation spaces are continua, a full-support conjecture may assign probability zero to any single contingency; what full support does rule out is the following.
\begin{lemma}\label{lem:WD_elimination}
  Fix $k \geq 1$, $i \in N$, and $\theta_{i}$, and let $m_{i}, m_{i}' \in M_{i}$. Suppose that, with $\mathbf{v}' := (v_{i},v'_{-i})$, we have $\pi_{i}(m_{i}',m_{-i};\mathbf{v}') \geq \pi_{i}(m_{i},m_{-i};\mathbf{v}')$ for every $(m_{-i},v'_{-i}) \in F_{i}^{k-1}(\theta_{i})$, with strict inequality on a nonempty relatively open subset of $F_{i}^{k-1}(\theta_{i})$. Then $m_{i} \notin T_{i}^{k}(\theta_{i})$.
\end{lemma}

We propose a simple mechanism that is weakly robustly efficient when there is one agent who is connected to everyone. This mechanism applies to certain networks where it is known by Theorem \ref{theo:robust_iff} that a robustly efficient mechanism does not exist. \vspace{2mm}

\noindent{\bf The one-center mechanism}\\
Let agent 1 be the central agent who is connected to everyone. Let $h_{i}(m_{i})$ and $\hat{h}_{i}(m_{i})$ be as in the mechanism in Section~\ref{sec:rob_eff}.
Let the allocation rule $g(\cdot)$ be the following.
\begin{itemize}
\item If $h_{1}(m_{1}) = 1$, then 
\begin{itemize}
\item if, for all $i \neq 1$, $m_{i1} \geq m_{ii}$, then $g(\mathbf{m}) = \mathbf{e}_{1}$,
\item otherwise, $g(\mathbf{m}) = \gamma^{1\hat{h}_{1}(m_{1})}$.
\end{itemize}
\item Otherwise, 
\begin{itemize}
\item if, for all $i \neq 1$, $m_{i1} \geq m_{ii}$, then $g(\mathbf{m}) = \mathbf{e}_{1}$
\item otherwise, $g(\mathbf{m}) = \mathbf{e}_{h_{1}(m_{1})}$
\end{itemize}
\end{itemize}\vspace{4mm}

To allocate the entire good to the central agent, all others must report that he is the highest. If some agent instead claims to outrank the center, the mechanism gives the allocation of $\gamma^{ij}$ such that $i$ and $j$ are the center and the highest ranked non-central agent by the center. This makes misreporting unattractive for the central agent: if he observes a higher-valued agent, he prefers that agent to receive the full good rather than a split allocation, while truthful reporting is only weakly enforced because the center still receives the full good whenever all others rank him highest. As a result, we can show that this mechanism is weakly robustly efficient.
\begin{proposition}\label{prop:weak_robust}
The one-center mechanism is weakly robustly efficient. Hence, there exists a weakly robustly efficient mechanism if there is one agent who is connected to everyone.
\end{proposition}
We illustrate the mechanism and its incentive scheme in the example below by using the simple network of three agents.\vspace{2mm}
\begin{example}
Consider a network with \(N=\{1,2,3\}\), where agent 1 is connected to both agents 2 and 3, while agents 2 and 3 are not connected to each other. Table \ref{tab:3_agents} illustrates the allocation rule of the one-center mechanism.
\begin{table}
  \centering
  \begin{subtable}[t]{0.48\textwidth}
    \centering
    \begin{tabular}{*{4}{c|}}
      \multicolumn{2}{c}{} & \multicolumn{2}{c}{$m_3$}\\\cline{3-4}
      \multicolumn{1}{c}{} &  & $m_{31} \geq m_{33}$ & $m_{31} < m_{33}$ \\\cline{2-4}
      \multirow{2}*{$m_2$} & $m_{21} \geq m_{22}$ & $(1,0,0)$ & $(\beta_{12}, \beta_{21}, 0)$ \\\cline{2-4}
                           & $m_{21} < m_{22}$   & $(\beta_{12}, \beta_{21}, 0)$ & $(\beta_{12}, \beta_{21}, 0)$ \\\cline{2-4}
    \end{tabular}
    \vspace{0.3em}
    \subcaption{$m_1: h_{1}(m_{1}) = 1, \hat{h}_{1}(m_{1}) = 2$}\label{subtab:SR1}
  \end{subtable}\hfill
  \begin{subtable}[t]{0.48\textwidth}
    \centering
    \begin{tabular}{*{4}{c|}}
      \multicolumn{2}{c}{} & \multicolumn{2}{c}{$m_3$}\\\cline{3-4}
      \multicolumn{1}{c}{} &  & $m_{31} \geq m_{33}$ & $m_{31} < m_{33}$ \\\cline{2-4}
      \multirow{2}*{$m_2$} & $m_{21} \geq m_{22}$ & $(1,0,0)$ & $(\beta_{13}, 0, \beta_{31})$ \\\cline{2-4}
                           & $m_{21} < m_{22}$   & $(\beta_{13}, 0, \beta_{31})$ & $(\beta_{13}, 0, \beta_{31})$ \\\cline{2-4}
    \end{tabular}
    \vspace{0.3em}
    \subcaption{$m_1: h_{1}(m_{1}) = 1, \hat{h}_{1}(m_{1}) = 3$}\label{subtab:SR2}
  \end{subtable}

  \vspace{0.5em}

  \begin{subtable}[t]{0.48\textwidth}
    \centering
    \begin{tabular}{*{4}{c|}}
      \multicolumn{2}{c}{} & \multicolumn{2}{c}{$m_3$}\\\cline{3-4}
      \multicolumn{1}{c}{} &  & $m_{31} \geq m_{33}$ & $m_{31} < m_{33}$ \\\cline{2-4}
      \multirow{2}*{$m_2$} & $m_{21} \geq m_{22}$ & $(1,0,0)$ & $(0, 1, 0)$ \\\cline{2-4}
                           & $m_{21} < m_{22}$   & $(0, 1, 0)$ & $(0, 1, 0)$ \\\cline{2-4}
    \end{tabular}
     \vspace{0.3em}
    \subcaption{$m_1: h_{1}(m_{1}) = 2$}\label{subtab:PR1}
  \end{subtable}\hfill
  \begin{subtable}[t]{0.48\textwidth}
    \centering
    \begin{tabular}{*{4}{c|}}
      \multicolumn{2}{c}{} & \multicolumn{2}{c}{$m_3$}\\\cline{3-4}
      \multicolumn{1}{c}{} &  & $m_{31} \geq m_{33}$ & $m_{31} < m_{33}$ \\\cline{2-4}
      \multirow{2}*{$m_2$} & $m_{21} \geq m_{22}$ & $(1,0,0)$ & $(0, 0, 1)$ \\\cline{2-4}
                           & $m_{21} < m_{22}$   & $(0, 0, 1)$ & $(0, 0, 1)$ \\\cline{2-4}
    \end{tabular}
     \vspace{0.3em}
    \subcaption{$m_1: h_{1}(m_{1}) = 3$}\label{subtab:PR2}
  \end{subtable}
  \caption{The allocation rule in the one-center mechanism for the three-agent case, where agent 1 is connected to agents 2 and 3, but agent 2 and 3 are not connected to each other. Each panel represents the allocation table given agent 1's message.}\label{tab:3_agents}
\end{table}

The example shows how the iterative elimination works. First, the center eliminates weakly dominated messages. If agent 1 is the highest, then the message ranking himself first weakly dominates messages that rank another agent first; if instead some agent \(i\) is higher than agent 1, then the message ranking \(i\) first weakly dominates the others. Thus, after the first round, only messages of agent 1 that truthfully identify the highest agent survive.

Given this, the peripheral agents' incentives become straightforward. An agent who observes that his own valuation exceeds agent 1's weakly prefers to report this truthfully, while an agent who observes that agent 1 exceeds him cannot profitably overturn the allocation by misreporting. Hence the surviving message profiles yield the efficient allocation. This illustrates why weak robustness can be achieved with a single center: the weak elimination step removes the central agent's misleading reports first, after which the remaining agents' truthful reports are sustained.
\end{example}\vspace{3mm}

As the above example illustrates, the process of iterative elimination begins with the central agent, who always has a weakly dominated message based on his observations. This fact allows us to eliminate the possibility of the central agent choosing an untruthful message to falsely claim that he has the highest valuation. In the context of (strictly) robust efficiency, the central agent may hold a conjecture that he is certain that others will send untruthful messages. By definition of the set $T_{i}(\theta_{i})$, these beliefs are excluded, and therefore we can construct a weakly robustly efficient mechanism.

\section{Beyond Networks with Two Universally Connected Agents}\label{sec:general_networks}
Theorem \ref{theo:robust_iff} restricts the two-hub construction to networks with at least two agents connected to everyone. We now ask what can still be said when this condition fails. The natural idea is to partition the agents into smaller groups --- or, more generally, to select any collection of possibly overlapping, not-necessarily-exhaustive groups --- each of which satisfies Theorem \ref{theo:robust_iff}'s own requirement internally, and to run the two-center construction behind Theorem \ref{theo:robust_iff} (built on the two-agent mechanism of Lemma \ref{lem:two_agents}) separately within each group. This section formalizes that idea and gives a worst-case guarantee on how far the resulting allocation can be from efficient.\footnote{We present the construction for Theorem \ref{theo:robust_iff} for expositional purposes rather than a limitation of the idea itself, since the same covering logic applies to Proposition \ref{prop:weak_robust}'s weakly robust efficiency, covering the network by single-hub neighborhoods instead of two-hub blocks. We develop that analogue in the online appendix.} 

A group of agents should count as verifiable if the principal can apply, within that group alone, the same cross-checking device that proves Theorem \ref{theo:robust_iff}: two members who each observe everyone else in the group, whose reports about the rest of the group can be checked against each other exactly as in the two-center mechanism (Lemma \ref{lem:two_agents}'s device, applied to the whole group rather than just its two hubs).

\begin{definition}[Verifiable block]\label{def:verifiable_block}
A subset $S \subseteq N$ is \textit{verifiable} if either (i) $|S|=1$; or (ii) there exist two agents $i,j\in S$, $i\neq j$, called \textit{hubs} of $S$, such that $S \subseteq N_i$ and $S \subseteq N_j$. A collection $\mathcal B = \{B_1,\dots,B_K\}$ of subsets of $N$ is \textit{verifiable} if every $B_k$ is.\footnote{The sets in $\mathcal B$ need not be disjoint, and need not cover $N$.}
\end{definition}

Condition (ii) is Theorem \ref{theo:robust_iff}'s own requirement, restricted to the induced subnetwork $\mathbf G[B_k]$ (agents in $B_k$, with $\mathbf G$'s links between them): the hubs of $B_k$ can run the same two-center construction as Theorem \ref{theo:robust_iff} on $B_k$, identifying its highest-valuation member without input from outside the group. Condition (ii) always requires $|S|\geq2$, so condition (i) is what makes an isolated agent's own singleton block verifiable. Two verifiable blocks may also overlap, including sharing a hub, a case we return to below. Hubs must belong to $S$ itself because Lemma \ref{lem:two_agents}'s discipline runs through the hubs' own shares; an outside verifier with no stake in $S$ would have nothing to be disciplined by.\vspace{2mm}

\noindent\underline{\textit{Block mechanism}}\\
Fix a verifiable collection $\mathcal B = \{B_1,\dots,B_K\}$ and weights $y_1,\dots,y_K \geq 0$ with $\sum_k y_k \leq 1$, fixed independently of any messages. For each $k$ with $|B_k|\geq 2$, fix a pair of hubs $i_k,j_k \in B_k$ as in Definition \ref{def:verifiable_block}; only $i_k$ and $j_k$ send messages to block $k$, denoted $m_{i_k}^{B_k}$ and $m_{j_k}^{B_k}$, reporting what each observes about $B_k$, elicited independently across the blocks in which their sender serves as hub. For each $k$ with $|B_k|\geq 2$, let $\tilde g^{B_k}(\mathbf m^{B_k})\in[0,1]^{B_k}$, with $\sum_{i\in B_k}\tilde g_i^{B_k}(\mathbf m^{B_k})\leq1$, be the allocation Theorem \ref{theo:robust_iff}'s two-hub construction assigns on $B_k$ given the block's messages, with the hubs of $B_k$ playing the role Section \ref{sec:rob_eff}'s two centers play on all of $N$; if $|B_k|=1$, say $B_k=\{i\}$, set $\tilde g^{B_k}\equiv\mathbf e_i$. The block mechanism allocates $\mathbf x = \sum_k y_k\,\tilde g^{B_k}(\mathbf m^{B_k})$, extending each $\tilde g^{B_k}$ by zero outside $B_k$ and summing the shares of an agent who wins in more than one block. When $\mathcal B$ is a partition of $N$, every agent is a hub of at most one block; when two blocks share a hub, that agent submits two separate, block-specific reports, one for each block in which he serves as hub. We call the case $K=1$ the \textit{single-block mechanism}.\vspace{2mm}

For $\mathbf v$ and a block $B_k\in\mathcal B$, write $i^*_k(\mathbf v):=\argmax_{j\in B_k}v_j$ for the highest-valuation agent within $B_k$ --- the block's analogue of the highest agent $i^*(\mathbf v)$ of Section \ref{sec:model} --- and $i^*_B(\mathbf v)$ for the same object when a single block $B$, rather than a member of an indexed collection, is under discussion.

\begin{proposition}\label{prop:block_achievability}
For every $\mathbf v$ and every block $B_k$ with either $|B_k|=1$ or $y_k>0$, the block mechanism satisfies $\tilde g^{B_k}(\mathbf m^{B_k}) = \mathbf e_{i^*_k(\mathbf v)}$ in every rationalizable message profile.
\end{proposition}
Proposition \ref{prop:block_achievability} follows from separability: since $y_k$ is fixed independently of any message, $j$'s payoff is additively separable across blocks, and $j$ controls a message only where he serves as hub. Within such a block $B_k$, his problem is --- up to the positive scalar $y_k$ --- exactly the two-agent problem of Lemma \ref{lem:two_agents}, so Theorem \ref{theo:robust_iff}'s argument applies blockwise and pins down $i_k^*(\mathbf v)$; this also covers overlapping blocks for free, since distinct blocks use disjoint message sets and a hub's report to one cannot be checked against, or affect his payoff in, another.\footnote{The block mechanism's message space is correspondingly larger than Section \ref{sec:model}'s --- a hub belonging to several blocks sends one report per block --- but the elimination operators of Section \ref{sec:rob_eff} apply verbatim to it, block by block.} The restriction to $y_k>0$ is needed because a zero-weight block carries no payoff term and so disciplines no message; singleton blocks need none and are exempt.
The weights $(y_k)$ are fixed independently of messages by design: letting a hub split two blocks according to his own private comparison of their champions would rest that split on a report the mechanism cannot verify, the same degenerate-belief channel behind the necessity direction of Theorem \ref{theo:robust_iff}.

Proposition \ref{prop:block_achievability} pins down only the best agent within a block, not necessarily the best agent in $N$, so the block mechanism need not be fully efficient. What it lets us give instead is a worst-case guarantee on how far its allocation can be from efficient in any network. Because the paper's belief-free approach rules out any distributional assumption on $\mathbf v$, there is no distribution to average over: the guarantee must hold for every $\mathbf v$, not merely on average. Stated in cardinal terms it would be uninformative here.\footnote{Since valuations are unrestricted on $[0,1]$ beyond distinctness, the worst-case value loss $v_{i^*(\mathbf v)}-v_i$, for any $\mathbf v$ where the recipient $i\neq i^*(\mathbf v)$, equals $1$.} We therefore take rank as the ordinal alternative, which does have worst-case content for such mechanisms.

\begin{definition}\label{def:rank}
For $i\in N$ and $\mathbf v$, let $\mathrm{rank}(i;\mathbf v) := |\{j\in N: v_j>v_i\}|+1$, so $\mathrm{rank}(i^*(\mathbf v);\mathbf v)=1$. Let $n_1^* := \max\{|S| : S \text{ verifiable}\}$ denote the size of the largest verifiable block, the subscript recording that it is the relevant statistic for a mechanism concentrated on a single block.
\end{definition}

\begin{corollary}\label{cor:rank_bound}
Let $B^*$ be a verifiable block with $|B^*|=n_1^*$, and consider the single-block mechanism with $y_1=1$ on $B^*$. Then, for every $\mathbf v$ and in every rationalizable message profile, the recipient's rank satisfies $\mathrm{rank}(i^*_{B^*}(\mathbf v);\mathbf v) \leq n-n_1^*+1$; and there exists $\mathbf v$ for which this bound holds with equality.
\end{corollary}

Putting the full weight on the largest verifiable block, $B^*$, is what pins the worst-case rank at $n-n_1^*+1$ rather than something looser; the same argument, with $n_1^*$ replaced by any block size $s$, gives the general bound of $n-s+1$. The bound moves between the paper's two polar cases: when $n_1^*=n$, $\mathbf G$ itself is verifiable and the bound gives $\mathrm{rank}\leq 1$, recovering full efficiency as in Theorem \ref{theo:robust_iff}; when $n_1^*$ is small relative to $n$, the guarantee is correspondingly weak.

\begin{remark}\label{rem:multiblock_no_gain}
Corollary \ref{cor:rank_bound}'s result---concentrate all weight on the largest verifiable block---is a reasonable choice under this paper's assumptions. Since $\mathbf v$ is drawn from no assumed distribution, only a worst-case criterion can discipline the choice among mechanisms, and rank is the only one with content here. A cardinal alternative, e.g. weighting each agent's valuation by his allocated share, $\sum_i x_i v_i$, is uninformative as a criterion of efficiency here: because valuations are unrestricted on $[0,1]$, this quantity's worst case is $0$ for every mechanism and every allocation, so it cannot distinguish between them.
\end{remark}

Corollary \ref{cor:rank_bound} implies that, among single-block mechanisms, guaranteeing the best possible worst-case rank requires finding the largest verifiable block. This raises a natural operational question---how costly is it to identify that block? Finding $n_1^*$, and a maximizing block, need not require searching all $2^n$ subsets of $N$.

\begin{proposition}[Computing $n_1^*$]\label{prop:compute_n1star}
If $\mathbf G$ has no edges, $n_1^*=1$. Otherwise,
$$
n_1^* = \max_{(i',j')\,:\,g_{i'j'}=1} |N_{i'}\cap N_{j'}|,
$$
and any $(i,j) \in \argmax_{(i',j')\,:\,g_{i'j'}=1} |N_{i'}\cap N_{j'}|$ gives $N_i\cap N_j$ as a largest verifiable block, with hubs $i,j$.
\end{proposition}
Computing $n_1^*$ is tractable: it reduces to a single pass over the edges of $\mathbf G$, an $O(n\cdot|E|)$ computation with a naive neighbor-set comparison at each edge, or faster with standard triangle-listing algorithms from the graph algorithms literature --- polynomial in the size of the network, unlike a brute-force search over all $2^n$ subsets. The reason a fast algorithm is available at all is that $|N_i\cap N_j|$ for an edge $(i,j)$ has a graph-theoretic reading: since $N_i$ includes $i$ itself, $|N_i\cap N_j|$ equals $2$ plus the number of triangles through edge $ij$, a standard statistic in network analysis.

Note that Corollary \ref{cor:rank_bound} is established only within the class of block mechanisms; whether the same guarantee, or a better one obtained by spreading weight across several blocks as discussed in Remark \ref{rem:multiblock_no_gain}, is achievable by some other mechanism entirely remains open.

\section{Discussion}\label{sec:discussion}
\subsection{Generalization of the utility function}
In this section, we discuss how the utility specification can be generalized, with particular emphasis on relaxing the assumption regarding the externality parameter $\alpha_i$. We then characterize conditions under which the mechanisms introduced in the previous sections remain applicable.

The key requirement for the results in the baseline model to be valid is that the two-agent mechanism in Lemma~\ref{lem:two_agents} is robustly efficient, since both the two-center mechanism and the one-center mechanism build on this mechanism. Specifically, for a pair of agents $1,2$ with utilities $u_1(\mathbf{x};\mathbf{v})$ and $u_2(\mathbf{x};\mathbf{v})$, there must exist an allocation $\mathbf{y} = (y_1,y_2)$ such that
\begin{equation}\label{eq:condition-general}
\begin{aligned}
  v_{1} > v_{2} \Rightarrow u_{1}(\mathbf{y};\mathbf{v}) > u_{1}(\mathbf{e}_{2};\mathbf{v}) \text{ and } u_{2}(\mathbf{e}_{1};\mathbf{v}) > u_{2}(\mathbf{y};\mathbf{v})\\
  v_{2} > v_{1} \Rightarrow u_{1}(\mathbf{y};\mathbf{v}) < u_{1}(\mathbf{e}_{2};\mathbf{v}) \text{ and } u_{2}(\mathbf{e}_{1};\mathbf{v}) < u_{2}(\mathbf{y};\mathbf{v})
\end{aligned}
\end{equation}
Condition \eqref{eq:condition-general} guarantees that if agents coordinate on a ``reversed-ranking'' report profile---that is, a profile that would induce the mechanism to allocate the good to the lower-valuation agent---then the high-valuation agent has a strict incentive to deviate. To simplify the condition, for $i\neq j$, define
\begin{align*}
\Delta_{ij}(\mathbf{x};\mathbf{v})
:= u_i(\mathbf{x};\mathbf{v}) - u_i(\mathbf{e}_j;\mathbf{v}),
\end{align*}
which measures agent $i$'s utility gain from $\mathbf{x}$ relative to assigning the entire unit to agent $j$. Then, condition \eqref{eq:condition-general} for one agent can be written as follows.

\begin{definition}[Sign-matching]
For a given allocation profile $\mathbf{x}$, the relative utility $\Delta_{ij}(\mathbf{x};\mathbf{v})$ satisfies the \textit{Sign-matching property} if 
\[
(v_i - v_j)\Delta_{ij}(\mathbf{x};\mathbf{v}) > 0 \quad \text{for all } v_i \neq v_j.
\]
\end{definition}
The \textit{Sign-matching property} dictates that the relative utility $\Delta_{ij}(\mathbf{x};\mathbf{v})$ and the difference in valuations $(v_i - v_j)$ have the same sign.\footnote{This is the analogue of the typical single-crossing property, which is expressed as a positive cross-derivative $\frac{\partial^2 U}{\partial q \partial v_i} > 0$. It implies that the utility difference between the allocation $\mathbf{x}$ and $\mathbf{e}_{j}$ must be increasing in the valuation difference $v_{i} - v_{j}$.}
We need to check the existence of an allocation profile $\mathbf{y}$ such that $\Delta_{12}(\mathbf{y};\mathbf{v})$ and $\Delta_{21}(\mathbf{y};\mathbf{v})$ satisfies the \textit{Sign-matching property}.

\subsubsection*{Linear-in-valuations utility functions}
We suggest one generalization of the utility function shown below.
\begin{align*}
  u_{i}(\mathbf{x};\mathbf{v}) = \sum_{j \in N}A_{ij}(\mathbf{x})v_{j}
\end{align*}
with $A_{ij}(\mathbf{x})$ being differentiable, $\frac{\partial A_{ij}(\mathbf{x})}{\partial x_{j}} > 0$, and $A_{ij}(\mathbf{x}) = 0$ if $x_{j} = 0$. We call this class of utility functions \textit{linear-in-valuations} utility functions. By setting $A_{ii}(\mathbf{x}) = x_{i}$ and $A_{ij}(\mathbf{x}) = \alpha_{i}x_{j}$ for $j \neq i$, we can recover the utility function in the main model. Moreover, this class of utility functions includes the case where the externality that agent $i$ receives depends on the agent to whom the good is allocated, e.g. $A_{ij}(\mathbf{x}) = \alpha_{ij}x_{j}$. By using this class, we can prove the existence of an allocation profile such that the \textit{sign-matching property} is satisfied for both agents $1$ and $2$, enabling us to establish the generalized two-agent mechanism.
\begin{proposition}\label{prop:general-2-agent-existence}
  Under linear-in-valuation utility functions, there exists an allocation profile $\mathbf{y} = (y_{1}, y_{2})$ such that $\Delta_{12}(\mathbf{y};\mathbf{v})$ and $\Delta_{21}(\mathbf{y};\mathbf{v})$ satisfy the \textit{sign-matching property}, if
  \begin{equation}\label{eq:sign-matching-1}
    \begin{aligned}
      u_{1}(\mathbf{x};\mathbf{v}) &> u_{1}(\mathbf{e}_{2};\mathbf{v}), &&\text{ for every $\mathbf{x}$ such that $x_{1} > 0$ and $x_{2} = 1 - x_{1}$, and}\\
      u_{2}(\mathbf{x};\mathbf{v}) &> u_{2}(\mathbf{e}_{1};\mathbf{v}), &&\text{ for every $\mathbf{x}$ such that $x_{2} > 0$ and $x_{1} = 1 - x_{2}$,}
    \end{aligned}
  \end{equation}
  for any $\mathbf{v}$ such that $v_{1} = v_{2}$. 
\end{proposition}
Proposition \ref{prop:general-2-agent-existence} guarantees the general two-agent mechanism to be robustly efficient. The condition for this is not demanding: it says that when two agents have identical valuations, an agent strictly prefers the good to be allocated to himself rather than the other agent, given that there is no money burning. In a general environment with $n$ agents, we say that the environment satisfies strict self-preference if the condition \eqref{eq:sign-matching-1} is satisfied for all pairs of agents $(i,j)$.

However, to extend this result to $n$ agents and apply the same machinery as in the mechanism used for Theorem \ref{theo:robust_iff}, an additional assumption is needed.

\begin{definition}[Efficiency-aligned reallocation]
  Let $x_{0} := 1 - \sum_{l \in N} x_{l}$ denote the burned part of the good, and let $i^{*} = i^{*}(\mathbf{v})$. The environment satisfies \textit{efficiency-aligned reallocation} if, for every agent $i$ and every $k \in (N \cup \{0\}) \setminus \{i, i^{*}\}$,
  \begin{align*}
    u_{i}(\mathbf{x} + t(\mathbf{e}_{i^{*}} - \mathbf{e}_{k});\mathbf{v}) > u_{i}(\mathbf{x};\mathbf{v}) \quad \text{for all } \mathbf{x} \text{ and } t \in (0, x_{k}].
  \end{align*}
  That is, every agent strictly benefits whenever any part of the good, allocated or burned, is moved to the highest-valuation agent from anyone other than himself.
\end{definition}

\begin{proposition}\label{prop:robust-mechanism-general}
  Assume linear-in-valuations utility functions, and that the environment satisfies strict self-preference and efficiency-aligned reallocation. A robustly efficient mechanism exists if the network is such that there are at least two agents who are connected to everyone. Moreover, there exists a weakly robustly efficient mechanism if there is one agent who is connected to everyone.
\end{proposition}
Efficiency-aligned reallocation requires agents' preferences to be aligned with the designer's objective for every part of the good, not only for the good as a whole; this is needed because the mechanism of Theorem \ref{theo:robust_iff} allocates split and partial shares. It holds in the main model, where the gain is $\alpha_{i}t(v_{i^{*}} - v_{k}) > 0$ if $i \neq i^{*}$ and $t(v_{i^{*}} - \alpha_{i^{*}}v_{k}) > 0$ if $i = i^{*}$, with $v_{0} := 0$. The property captures the partial alignment of incentives between designer and agents, serving as the core feature that enables robust implementation.\footnote{In the next section, we consider an alternative objective for the designer, illustrating a case where this property is not satisfied.}

\subsection{Welfare-maximizing designer}
In the baseline model, the designer's objective is to allocate the good to the agent with the highest valuation. Although this goal is intuitively appealing across various motivating examples, the presence of heterogeneous externalities implies that such an allocation does not inherently maximize aggregate utility. Moreover, from the point of view of implementation theory, assigning the good to the highest-valuation agent represents merely one possible social choice function. Accordingly, this subsection explores an alternative objective for the designer---maximizing utilitarian welfare---and evaluates both the validity and the limitations of the main framework's results under this distinct criterion.

We maintain the same setting as the baseline model, with the exception that the designer now seeks to allocate the good to maximize aggregate welfare. Formally, an allocation profile $\mathbf{x}$ is considered efficient if it maximizes total welfare:
\begin{align*}
  W(\mathbf{x};\mathbf{v}) &= \sum_{i \in N}u_{i}(\mathbf{x};\mathbf{v}) = \sum_{i \in N}\bigg[(1 - \alpha_{i}) +\sum_{j \in N}\alpha_{j} \bigg]v_{i}x_{i}
\end{align*}  
Consequently, the designer aims to allocate the entire good to the agent $i$ who maximizes the welfare weight:
\begin{align*}
  w_{i} &:= \bigg[1 - \alpha_{i} + \sum_{j \in N}\alpha_{j}\bigg]v_{i}  
\end{align*}
We say that agent $i$ is the welfare-maximizing recipient if $i = \argmax_{i' \in N}w_{i'}$. Furthermore, we say that a mechanism is robustly welfare-maximizing if, for any $(\alpha_{i})_{i \in N}$ and any valuation profile $\mathbf{v}$, the allocation is $g(\mathbf{m}) = \mathbf{e}_{i}$ for all rationalizable message profiles $\mathbf{m} \in \prod_{j \in N}S_{j}(\theta_{j})$, where $i$ is the welfare-maximizing recipient.

A key departure from the baseline model is that the objectives of the designer and the agents are not necessarily aligned. Specifically, the agent targeted by the designer is not always the one who would maximize the utility of the other agents, conditional on those agents not receiving the good themselves.\footnote{If $\alpha_{i} = \alpha_{j} = \alpha$ for all $i,j$, then maximizing $W(\mathbf{x};\mathbf{v})$ is equivalent to allocating the good to the agent with the highest valuation.} This misalignment introduces a conflict of interest that complicates robust implementation. Indeed, we can show that a robustly welfare-maximizing mechanism cannot exist if there exists an agent who is not connected to everyone else in the network.

\begin{proposition}\label{prop:welfare-impossibility}
  If there exists an agent who is not connected to everyone, then a robustly welfare-maximizing mechanism does not exist.
\end{proposition}

This result not only highlights the difficulty of achieving a welfare-maximizing mechanism but also demonstrates that the existence of the robust mechanism established in Theorem \ref{theo:robust_iff} fundamentally relies on the partial alignment of incentives between the designer and the agents.

The intuition underlying this result closely mirrors that of the impossibility result in Theorem \ref{theo:robust_iff}. The proof exploits the fact that an agent who lacks information about certain peers may hold extreme beliefs, assuming the valuations of those unobserved agents to be exceptionally high. To illustrate why a single pair of mutually unobserved agents is sufficient to yield the impossibility, suppose agents 1 and 2 do not observe each other. Consider a valuation profile $\mathbf{v}$ in which agent 1 is welfare-maximizing, yet $v_{2} > v_{1}$. In this scenario, allocating the good to agent 2 is inefficient; therefore, under a robust mechanism, there must exist an agent who has a profitable deviation from this allocation regardless of their beliefs. First, agents who observe both agents 1 and 2 cannot possess such a deviation. Because they know that agent 2 generates greater externalities (since $v_2 > v_1$), they lack any incentive to deviate in a way that shifts the allocation from agent 2 to agent 1. Furthermore, agents who fail to observe either agent 1 or agent 2 might hold a belief that agent 2's valuation is higher than that of agent 1.\footnote{For any agent $i \neq 1$, holding a belief that $v_{2} > v_{1}$ is sufficient. For agent 1, the belief we need to get the impossibility is $\alpha_{1} v_{2} > v_{1}$; that is, agent 2's valuation is presumed to be sufficiently high that agent 1 prefers agent 2 to receive the good---thereby benefiting from the resulting spillover---rather than obtaining the good himself.} Consequently, no agent can be reliably counted upon to correct the misallocation across all possible belief structures. 
This impossibility breaks down, however, if agent 1 is able to observe agent 2. In that case, agent 1 knows with certainty that he is the true welfare maximizer, giving him a strict incentive to deviate from any profile that allocates the good to agent 2 in order to secure the good for himself.

\subsection{Relation to ex post incentive compatibility}
Ex post implementation is heavily studied in the literature as an alternative notion of robustness in mechanism design and implementation theory. This section applies this concept to our model and discusses its implications.

The definitions are based on \citet{bergemann2008ex} adapted to our model, where a strategy of agent $i$ is a function $s_{i}:\Theta_{i} \to M_{i}$. Given a mechanism $g(\mathbf{m})$, a strategy profile $\bold{s}^{*} = (s_{1}^{*},\cdots,s_{n}^{*})$ is an ex post equilibrium if for all $i,\mathbf{v}$ and $m_{i}$,
\begin{align*}
  u_{i}(g(\mathbf{s}^{*}(\theta));\mathbf{v}) \geq u_{i}(g(m_{i},s^{*}_{-i}(\theta_{-i}));\mathbf{v}) \text{ where } \theta := \theta(\mathbf{v}).
\end{align*}
We say that the mechanism $g(\mathbf{m})$ is ex post efficient if every ex post equilibrium $s^{*}$ of the game induced by the mechanism $g(\mathbf{m})$ is such that $g(s^{*}(\theta))$ allocates the good to the highest agent, i.e. $g(\mathbf{s}^{*}(\theta(\mathbf{v}))) = \mathbf{e}_{i^{*}}$. 

One important thing to notice is that the type $\theta_{i}$ of agent $i$ consists of the valuations of agents in his neighborhood. That is, given a strategy $s_{i}$, the message must be the same as long as the valuations of agent $i$ and his neighbors are the same, regardless of the valuations of agents outside of his neighborhood.

This argument highlights an important condition discussed in the literature for a mechanism to be ex post efficient: \textit{ex post incentive compatibility} (EPIC). It guarantees that the truthful message profile is an equilibrium, which is the condition for \textit{partial implementation}, i.e. that there exists an ex post equilibrium that is efficient.\footnote{Usually, EPIC is defined over the social choice function. In our case, the space of possible message profiles does not equal the space of possible type profiles. This is due to the fact that the space of message profiles is larger than the space of type profiles, since the space of message profiles includes all inconsistent message profiles. Therefore, unlike the usual setting, a social choice function cannot be seen as a mechanism, and hence the EPIC is defined with a given mechanism.}
\begin{definition}[\citet{bergemann2008ex}]
  A direct mechanism $g(\cdot)$ is ex post incentive compatible (EPIC) if for all $i$, $\mathbf{v}$, and $\theta_{i}'$, we have
  \begin{align*}
    u_{i}(g(\theta(\mathbf{v}));\mathbf{v}) \geq u_{i}(g(\theta_{i}', \theta_{-i}(\mathbf{v}));\mathbf{v})
  \end{align*}
\end{definition} 

\begin{proposition}\label{prop:ex-post-IC}
  There exists a direct mechanism $g$ with $g(\theta(\mathbf v)) = \mathbf e_{i^*(\mathbf v)}$ that is ex post incentive compatible, if and only if all agents have at least one neighbor.
\end{proposition}

The necessity of no isolated agent follows because an isolated agent can overreport his valuation without creating any inconsistency: since no other agent reports the isolated agent's valuation, the mechanism must treat the report as truthful and allocate accordingly. As there always exist valuation profiles where the isolated agent prefers obtaining the good himself, he has an incentive to exaggerate. Conversely, when every agent has at least one neighbor, the mechanism which burns the entire good when reports are inconsistent is EPIC, since each agent has an incentive to report truthfully given the truthful reports of the others.

This does not prove the existence of an ex post efficient mechanism when no agent is isolated. Full ex post implementation requires two conditions: \textit{ex post incentive compatibility} (EPIC) and \textit{ex post monotonicity} (EM), shown by \citet{bergemann2008ex} to be necessary and almost sufficient. We have already established EPIC; in our setting, EM is also satisfied regardless of the network structure, as it depends only on the utility function and the social choice function.

\begin{definition}[\citet{bergemann2008ex}]
  An allocation rule $g$ satisfies \textit{ex post monotonicity} (EM) if, for every valuation profile $\mathbf{v}$ and every deceptive strategy profile $\hat{\mathbf{s}}$ such that $\hat{\mathbf{s}}(\theta(\mathbf{v})) = \theta(\mathbf{v}')$ for some $\mathbf{v}'$ with $i^{*}(\mathbf{v}) \neq i^{*}(\mathbf{v}')$, there exist an agent $i$ and an allocation $\mathbf{y} \in X$ such that 
\begin{align}\label{eq:ex-post-mono-1}
  u_{i}\big(\mathbf{y}; \mathbf{v}\big) > u_{i}\big(g(\theta(\mathbf{v}')); \mathbf{v}\big)
\end{align}
and, for any valuation profile $\tilde{\mathbf{v}}$ satisfying $\theta_{-i}(\tilde{\mathbf{v}}) = \theta_{-i}(\mathbf{v}')$,
we have
\begin{align}\label{eq:ex-post-mono-2}
u_{i}\big(g(\theta(\tilde{\mathbf{v}})); \tilde{\mathbf{v}}\big) \geq u_{i}\big(\mathbf{y}; \tilde{\mathbf{v}}\big).
\end{align}
\end{definition}
This condition guarantees the existence of a whistle-blower when the message profile is deceptive but consistent (equation \eqref{eq:ex-post-mono-1}), while ensuring no agent has an incentive 
to deviate when the profile is truthful (equation \eqref{eq:ex-post-mono-2}).

To see that EM is satisfied in our context, let $g(\theta(\mathbf{v})) = \mathbf{e}_{i}$ and $g(\theta(\mathbf{v}')) = \mathbf{e}_{j}$ with $i \neq j$, so $i$ and $j$ are the highest agents under $\mathbf{v}$ and $\mathbf{v}'$, respectively. By efficiency-aligned reallocation, for any $k \neq j$, there exists $\mathbf{y}$ with $y_i > 0$ such that $u_k(\mathbf{y};\mathbf{v}) > u_k(g(\theta(\mathbf{v}'));\mathbf{v})$, satisfying \eqref{eq:ex-post-mono-1}, for instance by choosing $\mathbf{y}$ so that $y_{i} > 0$ and $y_{j} = 1 - y_{i}$. Conversely, for any $k \neq i$, if $j$ is the highest agent, then $u_k(\mathbf{e}_j;\mathbf{v}) \geq u_k(\mathbf{y};\mathbf{v})$, satisfying \eqref{eq:ex-post-mono-2}.

\begin{remark}
  Assume that $n \geq 3$. If the utility function is continuous and the environment satisfies \textit{efficiency-aligned reallocation}, then \textit{ex post monotonicity} holds for $g$ mentioned above. 
\end{remark}

\section{Concluding remarks}\label{sec:conclusion}
A feature of our mechanisms that the analysis leaves implicit is how little the designer needs to know. The two-center mechanism of Theorem \ref{theo:robust_iff} uses only the messages of the two agents connected to everyone, and its punishment split $\gamma^{12}$ depends only on their externality weights $\alpha_{1}$ and $\alpha_{2}$. In the aid-targeting application, a program coordinator therefore needs to find two community members who know every household, such as a village head and a health worker, and to know how much each of them cares about the aid reaching the neediest household. She needs nothing about the other households: neither whom they know, nor how much they care about accurate targeting. The mechanism of Proposition \ref{prop:weak_robust} reverses this trade-off. It needs only one well-informed member, but its splits $\gamma^{1j}$ involve the externality weight of every household $j$ that this member may name, so the coordinator must know every household's stake in accurate targeting. Relaxing robustness with respect to agents' beliefs is thus paid for with more knowledge of their preferences on the designer's side.

The same comparison has a cost in terms of the network itself. The sparsest network admitting a weakly robustly efficient mechanism is the star, with $n - 1$ links, whereas the sparsest one admitting a robustly efficient mechanism, in which two agents are linked to each other and to everyone else, has $2n - 3$. Full robustness to belief hierarchies thus costs exactly $n - 2$ additional links: those of a second agent who observes everyone. In the targeting application, this is the step from relying on one well-informed member to adding a second.

\section*{Funding}
This work has received funding from the French government under the ``France 2030'' investment plan managed by the French National Research Agency (reference: ANR-17-EURE-0020) and from the Excellence Initiative of Aix-Marseille University - A*MIDEX. It has also benefited from the support of the Independent Research Fund Denmark (grant no.~4260-00050B).

\section*{Conflict of interest}
The author declares no conflict of interest.

\section*{Use of generative AI}
During the preparation of this work the authors used Claude (Anthropic) in order to improve language and readability. After using this tool, the authors reviewed and edited the content as needed and take full responsibility for the content of the publication.

\printbibliography

\appendix

\section*{Appendix: Proofs}
\noindent\bold{Proof of Lemma \ref{lem:two_agents}:} Label the two agents 1 and 2, and w.l.o.g. let $v_{1} > v_{2}$. Let $m_i$ and $m_i'$ be such that $m_{i1} \geq m_{i2}$ and $m_{i1}' < m_{i2}'$ for $i \in \{1,2\}$. We will check that for agent 1, $m_{1}'$ is strictly dominated. When agent 2 sends $m_2$, we have $\pi_1(m_1,m_2) = v_1$ and $\pi_1(m'_1,m_2) = 0 < v_1$. 
When agent 2 sends $m_2'$, we have $\pi_1(m_1,m'_2) = \frac{\alpha_1(1 - \alpha_2)}{1 - \alpha_1 \alpha_2}v_1 + \frac{\alpha_1\alpha_2(1 - \alpha_1)}{1 - \alpha_1\alpha_2}v_2$ and $\pi_1(m'_1,m'_2) = \alpha_1 v_2$.
Hence,
\begin{align*}
\pi_1(m_1,m'_2) - \pi_1(m'_1,m'_2) &= \frac{\alpha_1(1 - \alpha_2)}{1 - \alpha_1 \alpha_2}v_1 + \frac{\alpha_1\alpha_2(1 - \alpha_1)}{1 - \alpha_1\alpha_2}v_2 - \alpha_1 v_2\\
&= \frac{\alpha_1(1 - \alpha_2)}{1 - \alpha_1 \alpha_2}(v_1 - v_2) > 0 \text{ if and only if }v_1 > v_2
\end{align*}
Therefore, $m_1'$ is strictly dominated.
Once $m_{1}'$ is eliminated, we can check that for agent 2, $m_{2}'$ is strictly dominated. We have $\pi_2(m_1,m_2) = \alpha_2 v_1$ and $\pi_2(m_1,m'_2) = \frac{\alpha_2(1 - \alpha_1)}{1 - \alpha_1 \alpha_2}v_2 + \frac{\alpha_1\alpha_2(1 - \alpha_2)}{1 - \alpha_1\alpha_2}v_1$.
Hence, we have
\begin{align*}
\pi_2(m_1,m_2) - \pi_2(m_1,m'_2) &= \alpha_2 v_1 -   \frac{\alpha_2(1 - \alpha_1)}{1 - \alpha_1 \alpha_2}v_2 - \frac{\alpha_1\alpha_2(1 - \alpha_2)}{1 - \alpha_1\alpha_2}v_1 \\
&= \frac{\alpha_2(1 - \alpha_1)}{1 - \alpha_1 \alpha_2}(v_1 - v_2) > 0 \text{ if and only if }v_1 > v_2
\end{align*} \quad $\qedsymbol$\\

\noindent\bold{Proof of Theorem \ref{theo:robust_iff}:} {\it Necessity}\\
We first prove the following lemma.
\begin{lemma}\label{lem:nec_1}
Let $\bold{v}$ and $\bold{w}$ be valuation profiles with corresponding type profile $\theta$ and $\lambda$. If $\lambda_i \notin S_{i}(\theta_i)$ for some $i$, then there exists $j \in N$ such that $u_{j}(g(m_{j},\lambda_{-j});\mathbf{v}) > u_{j}(g(\lambda_{j},\lambda_{-j});\mathbf{v})$ for some $m_j \in M_{j}$.
\end{lemma}
\begin{proof}
If $\lambda_i \notin S_{i}(\theta_i)$, then it implies that there exists a largest $k$ such that $\lambda_{i'} \in S^k_{i'}(\theta_{i'})$ for all $i' \in N$. Let $\hat{k}$ be such $k$. Thus, there exists an agent $j \in N$ such that $\lambda_{j}$ is eliminated at round $\hat{k} + 1$, implying that for any $\mu_{j} \in C_{j}(\theta_{j})$ such that 
$$
\mu_{j}(\{(m_{-j},v_{-j})\,|\,m_{-j} \in S_{-j}^{\hat{k}}(\theta_{-j}(\bold{v}))\}) = 1
$$
there exists $m_{j}$ such that 
$$
\int_{M_{-j}\times V_{-j}}u_{j}(g(m_{j},m_{-j});(v_{j},v_{-j}))d\mu_{j} > \int_{M_{-j}\times V_{-j}}u_{j}(g(\lambda_{j}, m_{-j});(v_{j},v_{-j}))d\mu_{j}
$$
Since $\lambda_{-j} \in S_{-j}^{\hat{k}}(\theta_{-j})$, we can take $\mu_{j}^*$ such that $\mu_{j}^*(\lambda_{-j},v_{-j}) = 1$. 
Then, there must exist $m_{j}^*$ such that 
\begin{align*}
\int_{M_{-j}\times V_{-j}}\!u_{j}(g(m_{j}^*,m_{-j});(v_{j},v_{-j}))d\mu^*_{j} &> \int_{M_{-j}\times V_{-j}}\!u_{j}(g(\lambda_{j},m_{-j});(v_{j},v_{-j}))d\mu^*_{j} \nonumber\\
u_{j}(g(m_{j}^*,\lambda_{-j});(v_{j},v_{-j})) &> u_{j}(g(\lambda_{j},\lambda_{-j});(v_{j},v_{-j}))
\end{align*}
The lemma is proved.
\end{proof}

We prove the statement by contradiction. Assume that $g$ is robustly efficient, and assume that there is at most one agent who is connected to everyone. If such agent exists, assume without loss of generality that agent 1 is such agent. Otherwise, take randomly one agent and name him agent 1. Assume that a true valuation profile $\bold{v}^*$, with corresponding type profile $\theta^*$, is such that $v_{2}^{*} > v_{1}^{*} > v_{i}^{*}$ for every other $i$, i.e. agent 2 has the highest valuation and agent 1 has the second highest valuation. Moreover, assume that $v_{1}^{*} > \alpha_{1}v_{2}^{*}$ and $\alpha_{i}v_{1}^{*} > v_{i}^{*}$ for every $i \neq 1,2$.

Let us take an agent $i \neq 1$. First, assume that $i$ does not observe agent $2$. Take $\bold{v}^i$ and $\bold{w}^i$, with corresponding type profiles $\theta^i$ and $\lambda^{i}$, such that 
\begin{itemize}
  \item $v^i_{j} = v^*_j$ for all $j \in N$
  \item $w^i_{j} = v^*_{j}$ for all $j \neq 2$
  \item $\alpha_{2}w^i_1 > w_2^i$
\end{itemize}
Moreover, take some constant $\overline{w}_2$, and set $w_2^i = \overline{w}_2$ for every such $i$. By taking such $\bold{v}^i$ and $\bold{w}^i$, we can see that $\theta^i_i = \lambda^i_i$, and hence $S_{i}(\lambda_{i}^i) = S_{i}(\theta_{i}^i)$. By Lemma \ref{lem:nec_1}, if $\lambda_{i}^{i} \notin S_{i}(\lambda_{i}^i)$, then there exists $j \in N$ such that $u_{j}(g(m_{j},\lambda_{-j});\mathbf{w}^{i}) > u_{j}(g(\lambda_{j},\lambda_{-j});\mathbf{w}^{i})$ for some $m_j \in M_{j}$. We show that this is not possible. Let us denote $\bold{x}^* = g(m_{j}^*,\lambda_{-j}^i)$ and $\bold{x}^w = g(\lambda_{j}^i,\lambda_{-j}^i)$. By definition of $g$, we have $\bold{x}^w = \mathbf{e}_{1}$.

First, assume that $j \neq 1,2$. Then, we have
\begin{align*}
  w_{j}^i x_{j}^* + \alpha_{j} \sum_{l \neq j} w_{l}^i x^*_l > \alpha_{j} w_{1}^i
\end{align*}
This is not possible since $\alpha_{j} w_{1}^i = \alpha_{j} v_{1}^{*} > v_{j}^{*} = w_{j}^i$ for every $j \neq 1,2$, and $\alpha_{2} w_{1}^i > w_{2}^i$.

Next, assume that $j = 1$, then $w_{1}^i x_{1}^* + \alpha_{1} \sum_{l \neq 1} w_{l}^i x^*_l > w_{1}^i$
which is clearly not possible. Finally, assume that $j = 2$, then 
\begin{align*}
  w_{2}^i x_{2}^* + \alpha_{2} \sum_{l \neq 2} w_{l}^i x^*_l > \alpha_{2}w_{1}^i
\end{align*}
which is not possible since $\alpha_{2} w_{1}^{i} > w_{2}^{i}$. Therefore, $\lambda_{i}^{i} \in S_{i}(\lambda_{i}^i) = S_{i}(\theta_{i}^i)$.\vspace{2mm}

Now, assume that $i$ observes agent 2. Take an agent $k(i) \notin N_{i}$ and let $k := k(i)$. Let $\bold{v}^i$ and $\bold{w}^i$ be as follows.
\begin{itemize}
  \item $v_j^i = v_j^*$ for all $j \neq k$
  \item $i^{*}(\mathbf{w}^{i}) = k$ and $\alpha_{j}w_{k}^{i} > w_{j}^{i}$ for all $j \neq k$
  \item $w^i_j = v^i_j$ for all $j \notin \{2,k\}$, and $w^i_2 = \overline{w}_{2}$ such that $v_{1}^* > \overline{w}_2$
  \item $\alpha_{2}(v_{k}^i - w_{k}^i) \geq v_2^i - w_2^i$
\end{itemize}
Then, $\theta_i^i = \theta_i^*$, and $\lambda_i^i = (\overline{w}_2,(v_{j}^*)_{j \in N_i \setminus \{2\}})$. %
In words, in both $\bold{v}^i$ and $\bold{w}^i$, the valuation of $k$ is the highest, and $\bold{w}^i$ is such that, comparing to $\bold{v}^i$, the valuation of agent $2$ decreases less than $\alpha_2(v_{k}^{i} - w_{k}^i)$.  We can see that for $\overline{w}_2$ close enough to $v_{2}^{*}$, there exist $v_{k}^i$ and $w_{k}^{i}$ which satisfy these conditions. Moreover, notice that the second and the fourth bullet points imply that $i^{*}(\mathbf{v}^{i}) = k$.

\begin{lemma}
$\lambda_i^i \in S_{i}(\theta_i^{i})$ for any $i \neq 1$ such that $i \in N_2$.
\end{lemma}
\begin{proof}
Assume by contradiction that there exists an agent $i \in N_2$ such that $\lambda^i_i \notin S_{i}(\theta^i_i)$. Then, by Lemma \ref{lem:nec_1}, we know that there exists $j \in N$ such that 
\begin{align}
u_{j}(g(m_{j}^*,\lambda_{-j}^i);\mathbf{v}^i) &> u_{j}(g(\lambda_{j}^i,\lambda_{-j}^i);\mathbf{v}^i) \label{eq:elim_1}
\end{align}
for some $m_j^* \in M_{j}$.
Let us denote $\bold{x}^* = g(m_{j}^*,\lambda_{-j}^i)$ and $\bold{x}^w = g(\lambda_{j}^i,\lambda_{-j}^i)$.

Let us assume first that $j \notin \{2,k\}$. 
By definition of the mechanism $g$, we have $\mathbf{x}^w = \mathbf{e}_{k}$.
Therefore, we have $u_{j}(\bold{x}^w;\mathbf{v}^i) = \alpha_{j} v_{k}^i$ and $u_{j}(\bold{x}^*;\mathbf{v}^i) = v_{j}^i x_{j}^* + \alpha_{j} \sum_{l \neq j} v_{l}^i x^*_l$. By (\ref{eq:elim_1}), we have
\begin{align}
v_{j}^i x_{j}^* + \alpha_{j} \sum_{l \neq j} v_{l}^i x^*_l - \alpha_{j} v_{k}^i > 0 \label{eq:elim_2}
\end{align}
Besides, by construction of $\mathbf{w}^{i}$, we have $w_{j}^i x_{j}^* + \alpha_{j} \sum_{l \neq j}w_{l}^i x_{l}^* \leq \alpha_{j}w_{k}^i$.
Since $w_{l}^i = v_{l}^{i}$ for all $l \notin \{2,k\}$, this inequality can be written as follows.
\begin{align}
v_{j}^{i} x_{j}^* + \alpha_{j}\bigg( \sum_{l \neq 2,j,k}v_{l}^i x_{l}^* + w_{2}^{i}x_{2}^* + w_{k}^{i}x_{k}^* \bigg)  - \alpha_{j}w_{k}^i \leq 0 \label{eq:epic}
\end{align}
Therefore, by \eqref{eq:elim_2} and \eqref{eq:epic}, we obtain $\alpha_{j}[(v_{2}^i - w_{2}^{i})x_2^* - (v_{k}^i - w_{k}^{i})(1 - x_{k}^*)] > 0$.
This cannot hold if $x_{k}^* = 1$. This is because $\sum_{l \in N}x_{l}^* \leq 1$, and then if $x_{k}^* = 1$, then $x_2^* = 0$, which makes above expression equal zero. Hence, let $x_{k}^* < 1$. Then, since $v_{k}^i - w_{k}^i \geq v_{2}^i - w_{2}^i > 0$, we have
\begin{align*}
\alpha_{j}[(v_{2}^i - w_{2}^{i})x_2^* - (v_{2}^i - w_{2}^{i})(1 - x_{k}^*)] &\geq \alpha_{j}[(v_{2}^i - w_{2}^{i})x_2^* - (v_{k}^i - w_{k}^{i})(1 - x_{k}^*)] > 0
\end{align*}
This yields $\alpha_{j}(v_{2}^i - w_{2}^{i})(x_2^* + x_{k}^* - 1) > 0$, which is a contradiction since it must be that $x_2^* + x_{k}^* \leq 1$ and $v_{2}^{i} > w_{2}^{i}$ by construction.

Next, assume that $j = 2$. Then, by the same argument, we obtain $(v_{2}^i - w_{2}^{i})x_2^* - \alpha_{2}(v_{k}^{i} - w_{k}^{i})(1 - x_{k}^*) > 0$, and we can confirm that $x^*_{k} \neq 1$. Since $\alpha_{2}(v_{k}^i - w_{k}^i) \geq v_{2}^i - w_{2}^{i} > 0$, we have
\begin{align*}
(v_{2}^i - w_{2}^{i})x^*_2 - (v_{2}^{i} - w_{2}^{i})(1 - x_{k}^*) &\geq (v_{2}^i - w_{2}^{i})x_2^* - \alpha_{2}(v_{k}^{i} - w_{k}^{i})(1 - x_{k}^*) > 0
\end{align*}
This yields $(v_{2}^i - w_{2}^{i})(x_{k}^* + x_{2}^* - 1) > 0$, which is a contradiction.

Next, assume that $j = k$. Then, by applying \eqref{eq:elim_1} with $j = k$, we obtain 
\begin{align*}
  v_{k}^{i}x_{k}^{*} + \alpha_{k}\sum_{l \neq k}v_{l}^{i}x_{l}^{*} > v_{k}^{i}
\end{align*}
which is not possible since $v_{k}^{i} > v_{l}^{i}$ for every $l \neq k$.
\end{proof}
Now, we know that $\lambda^i_i \in S_{i}(\theta_i^*)$ for all $i \neq 1$, with $\lambda^i_i = (\overline{w}_2,(v_{j}^*)_{j \in N_i \setminus \{2\}})$ such that $v_{1}^* > \overline{w}_2$ for $i \in N_2$, and $\lambda^i_i = (v_{j}^*)_{j \in N_i}$ for $i \notin N_2$. Let us define $\hat{\bold{v}} := (\hat{v}_{i})_{i \in N}$ such that 
\begin{align*}
\hat{v}_{i} = 
\begin{cases}
\overline{w}_2, &\text{ for } i = 2\\
\hfil v^*_i, &\text{ otherwise }
\end{cases}
\end{align*} 
and let us denote $\hat{\lambda}_i := \lambda_i^i$ for each $i \neq 1$. Then, we can show that $\hat{\lambda}_1 \in S_{1}(\theta_1^*)$ with $\hat{\lambda}_1 = (\overline{w}_2,(v_{j}^*)_{j \notin \{2\}})$. If $\hat{\lambda}_1 \notin S_{1}(\theta_1^{*})$, then it must be the case that there is some $k$ such that $\hat{\lambda}_1 \in S_{1}^k(\theta_1^*)$ and $\hat{\lambda}_1 \notin S_{1}^{k+1}(\theta_1^*)$. However, we can take $\mu_1^*$ which puts probability 1 to $m_{-1} = \hat{\lambda}_{-1}$, because $\hat{\lambda}_{-1} \in S_{-1}(\theta_{-1}^*)$. By taking $\mu_1^*$, we always have $\hat{\lambda}_1 \in \BR_1(\theta_1^*,\mu_{1}^*)$. This is because $g(\hat{\lambda}_1,\hat{\lambda}_{-1}) = \mathbf{e}_{1}$, and hence $\hat{\lambda}_1$ cannot be eliminated.

Therefore, $\hat{\lambda} \in S(\theta^{*})$, which is a contradiction to $g$ being robustly efficient, since agent 1 is not the highest in $\mathbf{v}^{*}$, but $g(\hat{\lambda}) = \mathbf{e}_{1}$. \vspace{3mm}

\noindent{\it Sufficiency}\\
We show that the two-center mechanism presented in Section~\ref{sec:rob_eff} is robustly efficient. Without loss of generality, assume that $v_1 > v_2$. The steps of the proof are as follows.
\begin{enumerate}
  \item We prove that for agent 1, reporting that agent 2 is the unique highest agent is eliminated.
  \item Given the first step, we prove that for agent 2, reporting the true highest agent is the only surviving strategy.
  \item Given the two steps above, we prove that for agent 1, reporting the true highest agent is the only surviving strategy.
\end{enumerate}
\textit{\underline{Part 1}}: Let $m_{1}$ be such that $h_{1} = 2$. We prove that $m_{1}$ is never a best response of agent 1. Since agent 1 observes every valuation, he knows whether $i^{*} = 1$, and we treat the two cases separately.\vspace{2mm}\\
\noindent\textbf{Case A: $i^{*} \neq 1$.}\\
Since $v_{1} > v_{2}$, we have $i^{*} \neq 1,2$. Let $m_{1}'$ be the truthful message of agent 1, so that $h_{1}(m_{1}') = i^{*}$ and $m_{12}'/m_{1i^{*}}' = v_{2}/v_{i^{*}}$. Against $h_{2} = 1$ and $h_{2} = q \neq 1,2$, rule 1 applies to both messages, and $g(m_{1}',m_{2})$ is obtained from $g(m_{1},m_{2})$ by moving half of the good from agent 2 to $i^{*}$, which agent 1 strictly prefers since $v_{i^{*}} > v_{2}$. Against $h_{2} = 2$, $m_{1}$ yields $\mathbf{e}_{2}$ by rule 3, while $m_{1}'$ gives $i^{*}$ the share $v_{2}/v_{i^{*}} + \epsilon_{1}$ by rule 4, so that $\pi_{1}(m_{1}',m_{2}) = \alpha_{1}(v_{2} + \epsilon_{1}v_{i^{*}}) > \alpha_{1}v_{2} = \pi_{1}(m_{1},m_{2})$. Hence $m_{1}$ is strictly dominated by $m_{1}'$.\vspace{2mm}\\
\noindent\textbf{Case B: $i^* = 1$}\\
We first show that any $m_{2}$ with $q := h_{2}(m_{2}) \neq 1,2$ is strictly dominated by any $m_{2}'$ with $h_{2}(m_{2}') = 1$. Against $h_{1} = 1$, $m_{2}$ gives $q$ a share smaller than 1 by rule 4, while $m_{2}'$ yields $\mathbf{e}_{1}$ by rule 3. Against $h_{1} \neq 1$, rule 1 applies to both messages, and $g(m_{1},m_{2}')$ is obtained from $g(m_{1},m_{2})$ by moving half of the good from $q$ to agent 1. Since $i^{*} = 1$ implies $v_{1} > v_{q}$, agent 2 strictly prefers $m_{2}'$ in every case. Hence agent 2's surviving messages satisfy $h_{2} \in \{1,2\}$. Given this, $m_{1}$ is strictly dominated by $m_{1}'$ with $h_{1}(m_{1}') = 1$: against $h_{2} = 1$, $m_{1}'$ yields $\mathbf{e}_{1}$ instead of $\frac{1}{2}\mathbf{e}_{1} + \frac{1}{2}\mathbf{e}_{2}$, and $v_{1} > \alpha_{1}v_{2}$; against $h_{2} = 2$, $m_{1}'$ yields $\gamma^{12}$ instead of $\mathbf{e}_{2}$, which agent 1 strictly prefers by Lemma \ref{lem:two_agents}, since $v_{1} > v_{2}$.\vspace{2mm}\\
\noindent\textit{\underline{Part 2}}: We show that, given Part 1, $m_{2}$ such that $h_{2}(m_{2}) = i^{*}$ is the only surviving message. To do this, we show that $m_2$ strictly dominates any other strategy. 

First, assume that $1 \neq i^{*}$, and let 
\begin{itemize}
  \item $m_2'$ be such that $h_{2}' := h_{2}(m_{2}') = 1$,
  \item $m_2''$ be such that $h_{2}'' := h_{2}(m_{2}'') = 2$,
  \item $m_2^{(3)}$ be such that $h_{2}^{(3)} := h_{2}(m_2^{(3)}) = i$ for some $i \neq 1,2,i^{*}$.
\end{itemize}
We show that $m_2$ yields strictly higher payoff than $m_2'$, $m_2''$, and $m_2^{(3)}$ for any of the two possible strategies of agent 1, which are $h_{1}(m_{1}) = 1$ and $h_{1}(m_{1}) \neq 1,2$.

For $h_{1}(m_{1}) = 1$, we have 
\begin{align*}
\pi_2(m_2,m_{-2}) &= \alpha_2 \left(\frac{m_{21}}{m_{2i^*}} + \epsilon \right) v_{i^*}\\
\pi_{2}(m_2', m_{-2}) &= \alpha_2 v_1\\
\pi_{2}(m_2'', m_{-2}) &= \beta_{21}v_{2} + \alpha_2 \beta_{12}v_1\\
\pi_2(m_2^{(3)},m_{-2}) &= \alpha_2 \left(\frac{m_{21}^{(3)}}{m_{2h_{2}^{(3)}}^{(3)}} + \epsilon \right) v_{h_{2}^{(3)}}
\end{align*}
By Lemma \ref{lem:two_agents} and by the fact that $v_{i^*} = \max_{j \in N}v_{j}$, with $m_2$ such that $m_{21} = v_{1}$ and $m_{2i^{*}} = v_{i^{*}}$, we can show that $\pi_2(m_2,m_{-2}) = \alpha_{2}(v_{1} + \epsilon v_{i^{*}}) > \alpha_{2}v_{1} > \beta_{21}v_{2} + \alpha_2 \beta_{12}v_1.$
Moreover, for a given $m^{(3)}_{2}$, take $m_{2}$ such that $m_{2i^{*}} = m^{(3)}_{2h_{2}^{(3)}}$ and $m_{21} = m^{(3)}_{21}$, so that we have 
$$\pi_2(m_2,m_{-2}) = \alpha_2 \bigg(\frac{m_{21}}{m_{2i^*}} + \epsilon \bigg) v_{i^*} = \alpha_2 \bigg(\frac{m_{21}^{(3)}}{m_{2h_{2}^{(3)}}^{(3)}} + \epsilon \bigg)v_{i^{*}} > \alpha_2 \bigg(\frac{m_{21}^{(3)}}{m_{2h_{2}^{(3)}}^{(3)}} + \epsilon \bigg)v_{h_{2}^{(3)}}$$
For $h_{1}(m_{1}) \neq 1,2$, we have 
\begin{align*}
\pi_2(m_2,m_{-2}) &= \frac{\alpha_2}{2} \left(v_{i^*} + v_{h_{1}} \right)\\
\pi_{2}(m_2', m_{-2}) &=  \frac{\alpha_2}{2} \left(v_{1} + v_{h_{1}} \right)\\
\pi_{2}(m_2'', m_{-2}) &= \alpha_2 \left( \frac{m_{12}}{m_{1h_{1}}} + \epsilon \right)v_{h_{1}}\\
\pi_2(m_2^{(3)},m_{-2}) &= \frac{\alpha_2}{2} (v_{h_{1}} + v_{h_{2}^{(3)}} ) 
\end{align*}
By the fact that $v_{i^*} = \max_{j \in N}v_{j}$, we can show that $\pi_2(m_2,m_{-2})$ is higher than any other case for some $m_{2}$.

Now, assume that $i^{*} = 1$, and let 
\begin{itemize}
  \item $m_2'$ be such that $h_{2}' := h_{2}(m_{2}') = 2$,
  \item $m_2''$ be such that $h_{2}'' := h_{2}(m_2'') = i$ for some $i \neq 1,2$.
\end{itemize}
and consider the two strategies of agent 1 which are
\begin{itemize}
\item $h_{1}(m_{1}) = 1$
\item $h_{1}(m_{1}) \neq 1,2$.
\end{itemize}
For $h_{1}(m_{1}) = 1$, we have
\begin{align*}
\pi_{2}(m_2, m_{-2}) &= \alpha_2 v_1\\
\pi_{2}(m_2', m_{-2}) &= \beta_{21}v_{2} + \alpha_2 \beta_{12}v_1\\
\pi_2(m_2'',m_{-2}) &= \alpha_2 \left(\frac{m_{21}''}{m_{2h_{2}''}''} + \epsilon \right) v_{h_{2}''}
\end{align*}
With the same logic as the case where $i^{*} \neq 1$, we obtain that $\pi_{2}(m_2, m_{-2})$ is larger than under the other two strategies. 

For $h_{1}(m_{1}) \neq 1,2$, we have 
\begin{align*}
\pi_{2}(m_2, m_{-2}) &=  \frac{\alpha_2}{2} \left(v_{1} + v_{h_{1}} \right)\\
\pi_{2}(m_2', m_{-2}) &= \alpha_2 \left( \frac{m_{12}}{m_{1h_{1}}} + \epsilon \right)v_{h_{1}}\\
\pi_2(m_2'',m_{-2}) &= \frac{\alpha_2}{2} (v_{h_{1}} + v_{h_{2}''} ) 
\end{align*}
Again, $\pi_{2}(m_2, m_{-2})$ is the largest.\vspace{2mm}\\
\textit{\underline{Part 3}} : Given that agent 2 plays $h_{2}(m_{2}) = i^*$, we show that $m_{1}$ such that $h_{1}(m_{1}) = i^*$ strictly dominates any other strategy.

First, assume that $i^{*} \neq 1$, and let $m_{1}'$ and $m_{1}''$ be such that $h_{1}(m_{1}') = 1$ and $h_{1}'' := h_{1}(m_{1}'') \neq 1,i^{*}$, respectively. We have
\begin{align*}
  \pi_{1}(m_{1},m_{-1}) &= \alpha_{1}v_{i^{*}}\\
  \pi_{1}(m_{1}',m_{-1}) &= \alpha_{1}\bigg(\frac{m_{21}}{m_{2i^{*}}}+\epsilon\bigg)v_{i^{*}}\\
  \pi_{1}(m_{1}'',m_{-1}) &= \frac{\alpha_{1}}{2}(v_{i^{*}} + v_{h_{1}''} )
\end{align*}
We can see that $\pi_{1}(m_{1},m_{-1})$ is the largest. Now, assume $i^{*} = 1$. Then, $\pi_{1}(m_{1},m_{-1}) = v_{1}$, which is larger than the payoff for any other allocation. \quad $\qedsymbol$\\

\noindent \bold{Proof of Lemma \ref{lem:WD_elimination}:} Let $U$ be the nonempty relatively open set on which the inequality is strict, and let $\mu_{i}$ satisfy the conditions in the definition of $T_{i}^{k}(\theta_{i})$. Under $\mu_{i}$, the expected payoff of $m_{i}'$ exceeds that of $m_{i}$ by the integral over $F_{i}^{k-1}(\theta_{i})$ of a nonnegative function that is strictly positive on $U$, and $\mu_{i}(U\mid\theta_{i}) > 0$. Hence $m_{i} \notin \BR_{i}(\theta_{i},\mu_{i})$ for every such $\mu_{i}$. \quad $\qedsymbol$\\

\noindent \bold{Proof of Proposition \ref{prop:weak_robust}:} Assume that agent 1 is connected to everyone. Let $h_{1} := h_{1}(m_{1})$ be defined as in the mechanism. 
We prove the statement in three parts.
\begin{enumerate}
  \item We prove that for agent 1, telling a lie about the highest agent is eliminated at the first round, i.e. $m_1 \notin T^{1}_{1}(\theta_1)$ for $m_{1}$ such that $h_{1}(m_{1}) \neq i^{*}$.
  \item For all peripheral agents $i$ whose valuation is higher than agent 1, telling that agent 1's valuation is higher than himself is eliminated at the second round, i.e. for agent $i$ such that $v_i > v_1$, $m_i \notin T^2_{i}(\theta_i)$ if $m_{i1} > m_{ii}$.
  \item For all peripheral agents $i$ whose valuation is lower than agent 1, telling that agent 1's valuation is lower than himself is eliminated at the third round whenever agent 1 is the highest agent, i.e. for agent $i$ such that $v_1 > v_i$, $m_i \notin T^3_{i}(\theta_i)$ if $m_{ii} > m_{i1}$ and $i^{*} = 1$; if $i^{*} \neq 1$, agent $i$'s message does not affect the allocation.
\end{enumerate}\vspace{2mm}

\noindent\textit{\underline{Part 1}} : Let $m_{1}$ be such that $h_{1}(m_{1}) \neq i^{*}$, and let $m_{1}'$ be such that $h_{1}(m_{1}') = i^{*}$ and, if $i^{*} = 1$, $\hat{h}_{1}(m_{1}') = h_{1}(m_{1})$. If every peripheral agent concedes, both messages yield $\mathbf{e}_{1}$. Otherwise, agent 1 strictly prefers the allocation under $m_{1}'$. If $i^{*} \neq 1$, it is $\mathbf{e}_{i^{*}}$, while $m_{1}$ yields either $\mathbf{e}_{p}$ with $v_{p} < v_{i^{*}}$, or $\gamma^{1j}$ with $j := \hat{h}_{1}(m_{1})$; in the latter case, by the computation following Lemma \ref{lem:two_agents}, $\pi_{1}(\gamma^{1j}) - \pi_{1}(\mathbf{e}_{i^{*}}) = \beta_{1j}(v_{1} - v_{j}) - \alpha_{1}(v_{i^{*}} - v_{j}) < 0$, since $v_{i^{*}} \geq v_{j}$, $v_{i^{*}} > v_{1}$, and $\beta_{1j} = \alpha_{1}(1 - \beta_{j1}) < \alpha_{1}$. If $i^{*} = 1$, it is $\gamma^{1p}$ with $p := h_{1}(m_{1})$, while $m_{1}$ yields $\mathbf{e}_{p}$, and Lemma \ref{lem:two_agents} applies. Since the set of $m_{-1}$ in which some peripheral agent contests is a nonempty open subset of $F_{1}^{0}(\theta_{1})$, $m_{1} \notin T^{1}_{1}(\theta_{1})$ by Lemma \ref{lem:WD_elimination}.\vspace{2mm}

\noindent\textit{\underline{Part 2}} : Assume that there is an agent $i$ with $v_i > v_1$. Then, by Part 1, he knows that $m_1 \notin T^1_1(\theta_1)$ if $h_{1}(m_{1}) \neq i^{*}$. Hence, $F^{1}_{i}(\theta_{i})$ only includes $(m_{-i},v'_{-i})$ such that $h_{1}(m_{1}) = i^{*}(v_{i},v'_{-i})$.

Let $m_i$ and $m_i'$ be such that $m_{i1} \geq m_{ii}$ and $m_{ii}' > m_{i1}'$. For $m_i$, we have $g(m_i,m_{-i}) = \mathbf{e}_{1}$ if $m_{j1} \geq m_{jj}$ for all $j \notin \{1,i\}$, and $g(m_i,m_{-i}) = \mathbf{e}_{i^*}$ otherwise.
For $m_i'$, we always have $g(m_i',m_{-i}) = \mathbf{e}_{i^{*}}$. Since $v_{i} > v_{1}$ implies $i^{*} \neq 1$, we have $\pi_{i}(\mathbf{e}_{i^{*}}) > \pi_{i}(\mathbf{e}_{1})$. Hence $m_{i}'$ does at least as well as $m_{i}$, and strictly better on the nonempty relatively open set of contingencies in which every $j \notin \{1,i\}$ reports $m_{j1} > m_{jj}$, so $m_{i} \notin T^{2}_{i}(\theta_{i})$ by Lemma \ref{lem:WD_elimination}.\vspace{2mm}

\noindent\textit{\underline{Part 3}} : Take an agent $i$ with $v_1 > v_i$. Let $m_{i}$ and $m_i'$ be such that $m_{i1} \geq m_{ii}$ and $m_{ii}' > m_{i1}'$ respectively. We prove that $m_i$ yields a higher payoff than $m_i'$ for any $(m_{-i},v_{-i}') \in F_{i}^{2}(\theta_{i})$. Here, the set of $v_{-i}'$ consistent with $\theta_{i}(\mathbf{v})$ can be divided into two cases, namely $i^{*}(\mathbf{v}') \neq 1$, and $i^{*}(\mathbf{v}') = 1$. 

Take the first case and let $i^{*} := i^{*}(\mathbf{v}')$. By Part 1 and Part 2, we know that any $m_{-i} \in T^{2}_{-i}(\theta_{-i}(\mathbf{v}'))$ satisfies $h_{1}(m_{1}) = i^{*}$ and $m_{i^{*}i^{*}} > m_{i^{*}1}$. Hence, $g(m_{i},m_{-i}) = \mathbf{e}_{i^*}$ for any message of $i$.

Take the second case. By part 1, we know that any $m_{-i} \in T^{2}_{-i}(\theta_{-i}(\mathbf{v}'))$ satisfies $h_{1}(m_{1}) = 1$. Take any of such message $m_{1}$ and let us define $\hat{h}_{1}(m_{1})$ as in the one-center mechanism. Take first $m_{-1,i}$ such that for all $j \neq 1,i$, $m_{j1} \geq m_{jj}$. Then, $g(m_{i}, m_{-i}) = \mathbf{e}_{1}$, and $g(m_{i}',m_{-i}) = \gamma^{1\hat{h}_{1}(m_{1})}$. Since $1$ is the highest, $\mathbf{e}_{1}$ yields a strictly higher payoff for agent $i$ than $g(m_{i}',m_{-i})$: by Lemma \ref{lem:two_agents} if $i = \hat{h}_{1}(m_{1})$, and otherwise because $\gamma^{1\hat{h}_{1}(m_{1})}$ splits at most the whole good between agent 1 and an agent with a lower valuation.

Take now $m_{-1,i}$ such that $m_{jj} > m_{j1}$ for some $j \neq 1,i$. Then, we have $g(m_{i},m_{-i}) = \gamma^{1\hat{h}_{1}(m_{1})}$ for any $m_{i}$.

Hence $m_{i}$ does at least as well as $m_{i}'$. If $i^{*}(\mathbf{v}) = 1$, it does strictly better on the nonempty relatively open set of contingencies with $i^{*}(\mathbf{v}') = 1$ in which every $j \notin \{1,i\}$ reports $m_{j1} > m_{jj}$, so $m_{i}' \notin T^{3}_{i}(\theta_{i})$ by Lemma \ref{lem:WD_elimination}. If $i^{*}(\mathbf{v}) \neq 1$, agent $i$'s message does not affect the allocation, which is $\mathbf{e}_{i^{*}(\mathbf{v})}$ by the first case. \quad $\qedsymbol$\\

\noindent \bold{Proof of Proposition \ref{prop:block_achievability}:}
Fix $j\in N$ and let $\mathcal K_j$ be the blocks where $j$ serves as hub. Since $\mathbf x=\sum_k y_k\,\tilde g^{B_k}(\mathbf m^{B_k})$ and $u_j(\cdot;\mathbf v)$ is linear, $j$'s payoff is additively separable across blocks:
$$u_j(\mathbf x;\mathbf v)=\sum_{k\in\mathcal K_j}y_k\,u_j(\tilde g^{B_k}(\mathbf m^{B_k});\mathbf v) + \sum_{k\notin\mathcal K_j}y_k\,u_j(\tilde g^{B_k}(\mathbf m^{B_k});\mathbf v)$$
with the $k$-th term depending only on $j$'s own coordinate $m_j^{B_k}$ and on $B_k$'s other hub's message; and $j$'s message space $\prod_{k\in\mathcal K_j}M_j^{B_k}$ is a literal product across these blocks. Moreover, the second sum is independent of $j$'s message, and hence is taken as a constant.
This implies that, against any conjecture, $j$'s best response decomposes into $|\mathcal K_j|$ independent problems, one per block: for $y_k=0$ every message is a best response, and for $y_k>0$ the problem is, up to the scalar $y_k$, exactly $j$'s problem in the standalone two-hub game Theorem \ref{theo:robust_iff} analyzes on $B_k$.

This decomposition holds at every round of elimination, so by induction on the round --- the base case being the full message space, the inductive step the coordinatewise argument above applied to conjectures supported on the previous round's (already-factored, by hypothesis) surviving sets of the other hubs --- $j$'s surviving set is the product, over blocks $k\in\mathcal K_j$ with $y_k>0$, of his surviving set in the standalone game on $B_k$. By Theorem \ref{theo:robust_iff}'s sufficiency proof, every rationalizable profile of hub messages to such a block resolves to $\mathbf e_{i_k^*(\mathbf v)}$; for $|B_k|=1$ this holds by definition of $\tilde g^{B_k}$. \quad $\qedsymbol$\\

\noindent \bold{Proof of Corollary \ref{cor:rank_bound}:}
By Proposition \ref{prop:block_achievability}, the single-block mechanism on $B^*$ delivers the entire good to $i^*_{B^*}(\mathbf v)$ in every rationalizable message profile. Any agent with valuation exceeding $v_{i^*_{B^*}(\mathbf v)}$ must lie outside $B^*$, since $i^*_{B^*}(\mathbf v)$ is the maximizer within $B^*$; there are at most $|N\setminus B^*| = n-n_1^*$ such agents, giving $\mathrm{rank}(i^*_{B^*}(\mathbf v);\mathbf v) - 1 \leq n-n_1^*$. For attainment, take any $\mathbf v$ in which the $n-n_1^*$ highest valuations in $N$ are held by the agents in $N\setminus B^*$ (in any order) and the remaining, lower valuations are held by $B^*$. Then $i^*_{B^*}(\mathbf v)$ has exactly $n-n_1^*$ higher-valued agents, so $\mathrm{rank}(i^*_{B^*}(\mathbf v);\mathbf v) = n-n_1^*+1$.\quad $\qedsymbol$\\

\noindent \bold{Proof of Proposition \ref{prop:compute_n1star}:}
($\leq$) Let $S$ be verifiable with $|S|>1$, with hubs $i,j\in S$, $i\neq j$, so $S\subseteq N_i$ and $S\subseteq N_j$. Since $j\in S\subseteq N_i$ and $j\neq i$, this gives $g_{ij}=1$. Thus $S\subseteq N_i\cap N_j$, giving $|S|\leq|N_i\cap N_j|\leq\max_{(i,j):g_{ij}=1}|N_i\cap N_j|$. Since $S$ was an arbitrary verifiable set with $|S|>1$, this bound applies in particular to the maximizer defining $n_1^*$, whose size is strictly larger than $1$ since $\mathbf G$ has an edge; this gives $n_1^*\leq\max_{(i,j):g_{ij}=1}|N_i\cap N_j|$.

($\geq$) Let $(i,j) \in \argmax_{(i',j')\,:\,g_{i'j'}=1} |N_{i'}\cap N_{j'}|$. Since $g_{ij}=1$, we have $i\in N_j$ and $j\in N_i$. Hence, together with the fact $i\in N_i$ and $j\in N_j$, we have $i,j\in N_i\cap N_j$. The set $S:=N_i\cap N_j$ satisfies $S\subseteq N_i$ and $S\subseteq N_j$ by construction, so $S$ is verifiable with hubs $i,j$, giving $n_1^*\geq|N_i\cap N_j| = \max_{(i',j')\,:\,g_{i'j'}=1} |N_{i'}\cap N_{j'}|$.\quad $\qedsymbol$\\

\noindent \bold{Proof of Proposition \ref{prop:general-2-agent-existence}:} 
We need to show the existence of $\mathbf{y}$ such that $\Delta_{12}(\mathbf{y};\mathbf{v})$ and $\Delta_{21}(\mathbf{y};\mathbf{v})$ satisfy the \textit{sign-matching property}. If such $\mathbf{y}$ exists, then by the continuity of the utility function, it should satisfy
\begin{align}\label{eq:zero-intercept}
  v_{1} = v_{2} \quad \Longleftrightarrow \quad \Delta_{12}(\mathbf{y};\mathbf{v}) = \Delta_{21}(\mathbf{y};\mathbf{v}) = 0
\end{align}
We first prove the following lemma.
  \begin{lemma}\label{lem:zero-intercept}
    There exists an allocation profile $\mathbf{y}$ such that \eqref{eq:zero-intercept} holds, if,  whenever $v_{1} = v_{2}$, for every $\mathbf{x}$ such that $x_{1} > 0$ and $x_{2} = 1 - x_{1}$, we have $u_{1}(\mathbf{x};\mathbf{v}) > u_{1}(e_{2};\mathbf{v})$, and for every $\mathbf{x}$ such that $x_{2} > 0$ and $x_{1} = 1 - x_{2}$, we have $u_{2}(\mathbf{x};\mathbf{v}) > u_{2}(e_{1};\mathbf{v})$.
  \end{lemma}
  \begin{proof}
  The condition \eqref{eq:zero-intercept} for $\Delta_{12}(\mathbf{y},\mathbf{v})$ can be written as follows.
  \begin{align*}
    A_{11}(\mathbf{y})v_{1} + A_{12}(\mathbf{y})v_{2} = A_{12}(\mathbf{e}_{2})v_{2} \quad \Longleftrightarrow \quad v_{1} = v_{2} 
  \end{align*}
  By setting $v_{1} = v_{2} = v$, we obtain $A_{11}(\mathbf{y})v + A_{12}(\mathbf{y})v = A_{12}(e_{2})v$. Same thing applies for $\Delta_{21}(\mathbf{y};\mathbf{v})$, and we obtain the two following equations.
  \begin{align}
    A_{11}(\mathbf{y}) + A_{12}(\mathbf{y}) &= A_{12}(\mathbf{e}_{2})\label{eq:cond-zero-1}\\
    A_{22}(\mathbf{y}) + A_{21}(\mathbf{y}) &= A_{21}(\mathbf{e}_{1})\label{eq:cond-zero-2}
  \end{align}
  We prove the existence of $\mathbf{y}$ which satisfies \eqref{eq:cond-zero-1} and \eqref{eq:cond-zero-2}. Define a function $F(s,t) = (F_{1}(s,t), F_{2}(s,t))$ as follows.
  \begin{align*}
    F_{1}(s,t) := A_{11}(s, (1 - s)t) + A_{12}(s, (1 - s)t) - A_{12}(0,1)\\
    F_{2}(s,t) := A_{22}(s, (1 - s)t) + A_{21}(s, (1 - s)t) - A_{21}(1,0)
  \end{align*}
  where $s := y_{1}$, $t := \frac{y_{2}}{1 - y_{1}}$. Thus, \eqref{eq:cond-zero-1} and \eqref{eq:cond-zero-2} imply $F_{1}(s,t) = F_{2}(s,t) = 0$. By the Poincaré-Miranda theorem, it is sufficient to check that
  \begin{align}
    F_{1}(0,t) \leq 0 \text{ and }F_{1}(1,t) \geq 0 \text{ for all } t \in [0,1] \label{eq:cond-poinc-1}\\
    F_{2}(s,0) \leq 0 \text{ and }F_{2}(s,1) \geq 0 \text{ for all } s \in [0,1]\label{eq:cond-poinc-2}
  \end{align}
  \eqref{eq:cond-poinc-1} gives
  \begin{align*}
    A_{11}(0,t) + A_{12}(0,t) - A_{12}(0,1) = A_{12}(0,t) - A_{12}(0,1) \leq 0\\
    A_{11}(1,0) + A_{12}(1,0) - A_{12}(0,1) = A_{11}(1,0) - A_{12}(0,1) \geq 0
  \end{align*}
  The first inequality is true since $A_{12}(x_{1},x_{2})$ is non-decreasing in $x_{2}$. The second one is true when agent 1 prefers $(1,0)$ to $(0,1)$ whenever $v_{1} = v_{2}$. Moreover, \eqref{eq:cond-poinc-2} gives
  \begin{align*}
    A_{22}(s, 0) + A_{21}(s, 0) - A_{21}(1,0) = A_{21}(s, 0) - A_{21}(1,0) \leq 0\\
    A_{22}(s, 1 - s) + A_{21}(s, 1 - s) - A_{21}(1,0) \geq 0
  \end{align*}
  The first inequality is true since $A_{21}(x_{1},x_{2})$ is non-decreasing in $x_{1}$. The second inequality is true if $u_{2}(x_{1},x_{2};\mathbf{v}) \geq u_{2}(\mathbf{e}_{1};\mathbf{v})$ for every $(x_{1},x_{2})$ such that $x_{1} + x_{2} = 1$, whenever $v_{1} = v_{2}$.

  Moreover, every zero of $F$ satisfies $s \in (0,1)$ and $t > 0$, so that $y_{1} > 0$ and $y_{2} > 0$. Indeed, since $\partial A_{ij}/\partial x_{j} > 0$, the first inequality in \eqref{eq:cond-poinc-1} is strict for $t < 1$ and the first one in \eqref{eq:cond-poinc-2} is strict for $s < 1$, while $F_{1}(1,t) > 0$ and $F_{2}(0,1) > 0$ by the strict preferences assumed in the lemma.
  \end{proof}
  Now, by assuming that \eqref{eq:zero-intercept} holds, we prove the \textit{sign-matching property}. The sufficient condition for the \textit{sign-matching property} is $\Delta_{12}(\mathbf{y};\mathbf{v}) = k_{1}(v_{1} - v_{2})$ and $\Delta_{21}(\mathbf{y};\mathbf{v}) = k_{2}(v_{2} - v_{1})$ for some $k_{1}, k_{2} > 0$. By comparing the coefficient of $v_{1}$ and $v_{2}$ on the expression of $\Delta_{12}(\mathbf{y};\mathbf{v})$, we obtain $A_{11}(\mathbf{y}) - A_{11}(\mathbf{e}_{2}) = k_{1}$ and $A_{12}(\mathbf{y}) - A_{12}(\mathbf{e}_{2}) = - k_{1}$.
  The first equation implies $A_{11}(\mathbf{y}) = k_{1}$ since $A_{11}(\mathbf{e}_{2}) = 0$ by assumption. Moreover, by \eqref{eq:zero-intercept}, we have \eqref{eq:cond-zero-1}, and hence the second equation becomes
  \begin{align*}
    A_{12}(\mathbf{y}) - (A_{11}(\mathbf{y}) + A_{12}(\mathbf{y})) = - A_{11}(\mathbf{y}) = -k_{1}
  \end{align*}
  the same equation obtained by the first equation. By applying the same argument to $\Delta_{21}(\mathbf{y};\mathbf{v})$, we obtain $A_{22}(\mathbf{y}) = k_{2}$. Finally, $k_{1} = A_{11}(\mathbf{y}) > 0$ and $k_{2} = A_{22}(\mathbf{y}) > 0$, since $y_{1}, y_{2} > 0$ by the proof of Lemma \ref{lem:zero-intercept}, $A_{ii}(\mathbf{x}) = 0$ if $x_{i} = 0$, and $\partial A_{ii}/\partial x_{i} > 0$. \quad $\qedsymbol$\\

\noindent \bold{Proof of Proposition \ref{prop:robust-mechanism-general}:} Let $\mathbf{y}^{ij}$ be the allocation profile such that $y_{k}^{ij} = 0$ for every $k \neq i,j$, and   
  \begin{align*}
    u_{i}(y^{ij}_{i},y^{ij}_{j};\mathbf{v}) > u_{i}(\mathbf{e}_{j};\mathbf{v}) \text{ and } u_{j}(\mathbf{e}_{i};\mathbf{v}) > u_{j}(y^{ij}_{i},y^{ij}_{j};\mathbf{v})\Longleftrightarrow v_{i} > v_{j}\\
    u_{i}(y^{ij}_{i},y^{ij}_{j};\mathbf{v}) < u_{i}(\mathbf{e}_{j};\mathbf{v}) \text{ and } u_{j}(\mathbf{e}_{i};\mathbf{v}) < u_{j}(y^{ij}_{i},y^{ij}_{j};\mathbf{v})\Longleftrightarrow v_{i} < v_{j}
  \end{align*}
  \textit{Robustly efficient mechanism:} Take two agents, 1 and 2, who are connected to everyone.
Let the mechanism $g(\cdot)$ be the same as the one in the proof of Theorem \ref{theo:robust_iff}, except that if $h_{1} = 1$ and $h_{2} = 2$, then we set $g(m_{1},m_{2}) = \mathbf{y}^{12}$.
We prove that this mechanism is robustly efficient. Without loss of generality, assume that $v_1 > v_2$.\vspace{2mm}

The proof of the sufficiency part of Theorem \ref{theo:robust_iff} carries over, with the same messages eliminated in the same order. Each strict comparison made there is of one of the following types. (i) $\mathbf{y}^{ij}$, in place of $\gamma^{ij}$, against $\mathbf{e}_{i}$ or $\mathbf{e}_{j}$: this holds by the definition of $\mathbf{y}^{ij}$. (ii) Two allocations such that the preferred one is obtained from the other by moving part of the good, allocated or burned, to $i^{*}$ from someone other than the agent making the comparison: this holds by efficiency-aligned reallocation. (iii) A partial allocation $\sigma\mathbf{e}_{j}$ from rule 4 against an allocation that beats $\mathbf{e}_{j}$ by (i) or (ii): this holds since $u_{i}(\sigma\mathbf{e}_{j};\mathbf{v}) \leq u_{i}(\mathbf{e}_{j};\mathbf{v})$ by $\partial A_{ij}/\partial x_{j} > 0$. The only step that uses linearity beyond these comparisons is where a partial allocation $\sigma\mathbf{e}_{i^{*}}$ must beat a full allocation to another agent, namely in Case A of Part 1 against $h_{2} = 2$, and in Part 2 against $h_{1} = 1$ when $i^{*} \neq 1$. There, the truthful ratio is replaced by a reported ratio close enough to 1: since $u_{i}(\mathbf{e}_{i^{*}};\mathbf{v}) > u_{i}(\mathbf{e}_{j};\mathbf{v})$ for $j \neq i,i^{*}$ by efficiency-aligned reallocation and $u_{i}$ is continuous, $\sigma\mathbf{e}_{i^{*}}$ is then strictly better, and the ratio affects no other comparison.\vspace{2mm}

\noindent \textit{Weakly robustly efficient mechanism}: Let agent 1 be the central agent who is connected to everyone. Let $h_{i} := h_{i}(m_{i})$ and $\hat{h}_{i} := \hat{h}_{i}(m_{i})$ be defined the same as in the one-center mechanism. 

Let the mechanism be the same as the one-center mechanism, except in the case where $h_{1}(m_{1}) = 1$ and $m_{ii} > m_{i1}$ for some $i \neq 1$. In this case, we set $g(\mathbf{m}) = \mathbf{y}^{1\hat{h}_1}$.
The proof of Proposition \ref{prop:weak_robust} carries over, with the same messages eliminated on the same sets of contingencies, and with comparisons of types (i) and (ii) above in place of those relying on linearity. The only comparison of another type arises in part 1, when $i^{*} \neq 1$ and $h_{1}(m_{1}) = 1$: agent 1 then compares $\mathbf{y}^{1j}$, where $j := \hat{h}_{1}(m_{1})$, with $\mathbf{e}_{i^{*}}$. If $v_{j} > v_{1}$, then $u_{1}(\mathbf{y}^{1j};\mathbf{v}) < u_{1}(\mathbf{e}_{j};\mathbf{v}) \leq u_{1}(\mathbf{e}_{i^{*}};\mathbf{v})$ by (i) and (ii). If $v_{j} < v_{1}$, let $c := A_{1j}(\mathbf{e}_{j})$. By \eqref{eq:cond-zero-1} applied to the pair $(1,j)$, $u_{1}(\mathbf{y}^{1j};\mathbf{v}) = c v_{1} - A_{1j}(\mathbf{y}^{1j})(v_{1} - v_{j}) < c v_{1}$, while efficiency-aligned reallocation at profiles in which $v_{j}$ approaches $v_{i^{*}}$ gives $A_{1i^{*}}(\mathbf{e}_{i^{*}}) \geq c$, so that $u_{1}(\mathbf{e}_{i^{*}};\mathbf{v}) \geq c v_{i^{*}} > c v_{1}$.\quad $\qedsymbol$\\

\noindent \textbf{Proof of Proposition \ref{prop:welfare-impossibility}:} 
Take two agents, $1$ and $2$, and assume that they do not observe each other. We assume by contradiction that a mechanism $g$ robustly allocates to the welfare maximizer, that is, an agent $i$ such that $i = \argmax_{i' \in N}\big(1 - \alpha_{i'} + \sum_{j \in N}\alpha_{j}\big)v_{i'}$.
  
Take two valuation profiles, $\mathbf{v}$ and $\mathbf{v}'$ as follows. Let agent $1$ be the welfare maximizer under $\mathbf{v}$, and let $i^{*} = 2$. Moreover, let $\alpha_{i}v_{2} > v_{i}$ for all $i \neq 1,2$.
Let $2$ be the welfare maximizer under $\mathbf{v}'$. Moreover, let $v_{k}' = v_{k}$ for any $k \neq 2$, and let $\alpha_{1}v'_{2} > v_{1}$. Since $g$ is robustly welfare-maximizing, the good should be allocated to $2$ at $\mathbf{v}'$ with some message profile which survives the iterative eliminations. Let $\mathbf{m}$ be this message profile, and therefore we have $g(\mathbf{m}) = \mathbf{e}_{2}$.

On the other hand, since $\mathbf{e}_{2}$ is not welfare maximizing at $\mathbf{v}$, there must be an agent $i \in N$ such that, for any belief $b_{i} \in \Delta(V_{-i})$ consistent with his information, we have 
\begin{align*}
  \int_{V_{-i}}\pi_{i}(\mathbf{e}_{2};(v_{i},w_{-i}))db_{i}(w_{-i}) < \int_{V_{-i}}\pi_{i}(g(m_{i}', m_{-i});(v_{i},w_{-i}))db_{i}(w_{-i}) \text{ for some }m_{i}' 
\end{align*}
First, we prove that such $i$ cannot be an agent $i \neq 1$. Assume by contradiction that an agent $i \neq 1$ is such agent, and take a belief such that $b_{i}(v_{-i}) = 1$, i.e. the agent has a correct belief. By letting $\mathbf{x}' = g(m_{i}', m_{-i})$, this implies
\begin{align*}
  \alpha_{i}v_{2} < x_{i}'v_{i} +\alpha_{i}\sum_{j \neq i}x_{j}'v_{j} \Rightarrow \alpha_{i}(1 - x_{2}')v_{2} < x_{i}'v_{i} + \alpha_{i}\sum_{j \neq 1,2,i}x_{j}'v_{j} + \alpha_{i}x_{1}'v_{1}
\end{align*}
Since $\alpha_{i}v_{2} > v_{i}$, we have
\begin{align*}
  x_{i}'v_{i} + \alpha_{i}\sum_{j \neq 1,2,i}x_{j}'v_{j} + \alpha_{i}x_{1}'v_{1} < \alpha_{i}x_{i}'v_{2} + \alpha_{i}\sum_{j \neq 1,2,i}x_{j}'v_{2} + \alpha_{i}x_{1}'v_{2}
\end{align*}
Hence, we obtain 
\begin{align*}
  \alpha_{i}(1 - x_{2}')v_{2} < \alpha_{i}x_{i}'v_{2} + \alpha_{i}\sum_{j \neq 1,2,i}x_{j}'v_{2} + \alpha_{i}x_{1}'v_{2} \Rightarrow 1 < x_{i}' + \sum_{j \neq 1,2,i}x_{j}' + x_{1}' + x_{2}'
\end{align*}
This is not possible since $\sum_{i \in N}x'_{i} \leq 1$. 
This means that such agent should be agent 1. However, we prove that this is not possible.
Assume by contradiction that this is the case. Then, for any belief $b_{1} \in \Delta(V_{-1})$, we have
\begin{align*}
  \int_{V_{-1}}\pi_{1}(\mathbf{e}_{2};(v_{1},w_{-1}))db_{1}(w_{-1}) < \int_{V_{-1}}\pi_{1}(g(m_{1}', m_{-1});(v_{1},w_{-1}))db_{1}(w_{-1}) \text{ for some }m_{1}'
\end{align*}
Take a belief such that $b_{1}(w_{-1}) = 1$ where $w_{-1}$ is such that $w_{2} = v_{2}'$ and $w_{i} = v_{i}$ for all $i \neq 1,2$. This is the belief that agent 1 believes that agent 2's valuation is $v_{2}'$, and for any other agent $i$, the valuation is $v_{i}$, the correct one. This belief is consistent to the information of agent 1 since agent 1 does not observe agent 2. By taking this belief, we obtain
\begin{align*}
  \alpha_{1}v_{2}' < x_{1}'v_{1} + \alpha_{1}\sum_{j \neq 1,2}x'_{j} w_{j} + \alpha_{1}x_{2}'v_{2}' \Rightarrow \alpha_{1}v_{2}'(1 - x_{2}') &< x_{1}'v_{1} + \alpha_{1}\sum_{j \neq 1,2}x'_{j} v_{j}
\end{align*}
Since $\alpha_{1}v_{2}' > v_{1}$, and $v_{2}' > v_{2} > v_{i}$ for all $i \neq 1,2$, we obtain
\begin{align*}
  \alpha_{1}v_{2}'(1 - x_{2}') < x_{1}'\alpha_{1}v_{2}' + \alpha_{1}\sum_{j \neq 1,2}x'_{j} v_{2} \Rightarrow 1 - x_{2}' < x_{1}' + \sum_{j \neq 1,2}x'_{j}
\end{align*}
which is a contradiction since $\sum_{j \in N}x'_{j} \leq 1$. \quad $\qedsymbol$\\

\noindent \textbf{Proof of Proposition \ref{prop:ex-post-IC}}\\
($\Rightarrow$): Assume by contradiction that there is a direct mechanism $g$ with $g(\theta(\mathbf v)) = \mathbf e_{i^*(\mathbf v)}$ that satisfies EPIC under a network where there is an isolated agent. Let agent $1$ be such an isolated agent. 
Let $\mathbf{v}$ and $\mathbf{v}'$ be the valuation profiles as follows. $i^{*}(\mathbf{v}) = 1$, and $v_{i}' = v_{i}$ for all $i \neq 1$, $i^{*}(\mathbf{v}') = 2$, and $v_{1}' > \alpha_{1}v_{2}'$. Let $\theta := \theta(\mathbf{v})$ and $\theta' := \theta(\mathbf{v}')$. Since $g(\theta(\mathbf v)) = \mathbf e_{i^*(\mathbf v)}$, we have $g(\theta) = \mathbf e_1$ and $g(\theta') = \mathbf e_2$.

Since for all $i \neq 1$, we have $v_{i} = v_{i}'$ and they are not connected to agent $1$, for all agent $i \neq 1$, we have $\theta_{i} = \theta_{i}'$. Hence, we have $g(\theta_{1},\theta_{-1}') = g(\theta_{1},\theta_{-1}) =  \mathbf{e}_{1}$ and $g(\theta_{1}', \theta_{-1}') = g(\theta_{1}', \theta_{-1}) = \mathbf{e}_{2}$.
Besides, since $v_{1}' > \alpha_{1}v_{2}'$, we have
\begin{align*}
  u_{1}(g(\theta_{1}', \theta_{-1}');\mathbf{v}') = \alpha_{1}v_{2}' < v_{1}' = u_{1}(g(\theta_{1}, \theta_{-1});\mathbf{v}') = u_{1}(g(\theta_{1}, \theta_{-1}');\mathbf{v}')
\end{align*}
This implies that $g$ is not EPIC.\vspace{2mm}

\noindent ($\Leftarrow$): Assume that everyone has at least one neighbor and let the mechanism $g$ be as follows. 
\begin{itemize}
  \item $g(\theta(\mathbf{v})) = \mathbf{e}_{i^{*}(\mathbf{v})}$.
  \item $g(\mathbf{m}) = \mathbf{0}$ otherwise.
\end{itemize}
In words, the mechanism allocates to the highest agent reported by the message profile if the message profile is consistent, and burns the whole good if the message profile is inconsistent. We show that this is EPIC. Take an agent $i \neq i^{*}$, and from the mechanism, we have $u_{i}(g(\theta(\mathbf{v})); \mathbf{v}) = \alpha_{i}v_{i^{*}} > 0 = u_{i}(g(\theta_{i}',\theta_{-i}(\mathbf{v})); \mathbf{v})$.
The message profile $(\theta_{i}',\theta_{-i}(\mathbf{v}))$ is necessarily inconsistent since there is at least one agent who is connected to $i$. For agent $i^{*}(\mathbf{v})$, same argument holds.
\quad $\qedsymbol$\vspace{3mm}

\end{document}